\documentclass[a4paper,UKenglish]{lipics-v2016}

\usepackage{microtype}

\newcommand{\nc}{\newcommand}
\nc{\rnc}{\renewcommand} \nc{\nev}{\newenvironment}

\nc{\W}[1]{\text{$\textup{W}[#1]$}}

\newlength{\probwidth}
\nc{\prob}[3][9]{
\begin{center}
  \normalfont\fbox{
   \begin{tabular}[t]{
     rp{#1cm}}\textit{Instance:}&#2. \\
     \textit{Problem:}&#3
   \end{tabular}}
\end{center}}

\nc{\pprob}[4][9]{
\begin{center}
   \normalfont\fbox{
    \begin{tabular}[t]{
     rp{#1cm}}\textit{Instance:}&#2. \\
     \textit{Parameter:}&#3. \\
     \textit{Problem:}&#4
   \end{tabular}}
\end{center}}

\nc{\nprob}[4][9]{
\begin{center}
  \normalfont\fbox{

\addtolength{\probwidth}{#1cm}\parbox{\probwidth}{\textsc{#2}\\\hspace*{1.5em}
     \begin{tabular}[t]{
      rp{#1cm}}\textit{Instance:}&#3. \\
      \textit{Problem:}&#4
     \end{tabular}}}
\end{center}}

\nc{\npprob}[5][9]{
\begin{center}
  \normalfont\fbox{

\addtolength{\probwidth}{#1cm}\parbox{\probwidth}{\textsc{#2}\\\hspace*{1.5em}
    \begin{tabular}[t]{
     rp{#1cm}}\textit{Instance:}&#3. \\
     \textit{Parameter:}&#4. \\
     \textit{Problem:}&#5
    \end{tabular}}}
\end{center}}

\nc{\nppxrob}[5][9]{ \normalfont\fbox{

\addtolength{\probwidth}{#1cm}\parbox{\probwidth}{\textsc{#2}\\\hspace*{1.5em}
   \begin{tabular}[t]{
    rp{#1cm}}\textit{Instance:}&#3. \\
    \textit{Parameter:}&#4. \\
    \textit{Problem:}&#5
   \end{tabular}}}}

\nc{\nppprob}[5][4]{
\begin{center}
  \normalfont\fbox{

\addtolength{\probwidth}{#1cm}\parbox{\probwidth}{\textsc{#2}\\\hspace*{1.5em}
    \begin{tabular}[t]{
     rp{#1cm}}\textit{Instance:}&#3. \\
     \textit{Parameter:}&#4. \\
     \textit{Problem:}&#5
    \end{tabular}}}
\end{center}}

\nc{\dotcup}{\;\dot\cup\;}

\title{Fixed-parameter Approximability of Boolean MinCSPs\footnote{This work was supported by the European Research Council (ERC) starting grant "PARAMTIGHT: Parameterized complexity and the search for tight complexity results" (reference 280152) and OTKA grant NK105645. The second author was supported by NSERC. The third author was  supported by the
JSPS KAKENHI Grant (JP16H07409) and the JST ERATO Grant (JPMJER1201)
of Japan.}}

\author[1]{\'Edouard Bonnet}
\author[2]{L\'{a}szl\'{o} Egri}
\author[3]{Bingkai Lin}
\author[4]{D\'{a}niel Marx}
\affil[1]{Université de Lyon (COMUE), CNRS, ENS de Lyon, Université Claude-Bernard Lyon 1, LIP, France\\
	\texttt{edouard.bonnet@ens-lyon.fr}}
\affil[2]{Department of Mathematics and Computer Science, Indiana State University, Terre Haute, United States\\
	\texttt{laszlo.egri@mail.mcgill.ca}}
\affil[3]{National Institute of Informatics \\
  \texttt{lin@nii.ac.jp}}
\affil[4]{Institute for Computer Science and Control, Hungarian Academy of Sciences (MTA SZTAKI), Budapest, Hungary\\
  \texttt{dmarx@cs.bme.hu}}
\authorrunning{\'E. Bonnet, L. Egri, B. Lin, and D. Marx}
\Copyright{\'Edouard Bonnet, L\'{a}szl\'{o} Egri, Bingkai Lin, and D\'{a}niel Marx}

\usepackage{graphicx}
\usepackage{amsfonts}
\usepackage{amsmath}
\usepackage{latexsym}
\usepackage{amssymb}
\usepackage{color}
\usepackage{verbatim}
\usepackage{xspace}
\usepackage{nicefrac}

\newtheorem{proposition}[theorem]{Proposition}
\newtheorem{claim}[theorem]{Claim}

\usepackage{comment}

\DeclareMathOperator{\operatorClassNP}{W[P]}
\newcommand{\WP}{\ensuremath{\operatorClassNP}}

\DeclareMathOperator{\operatorClassFPT}{FPT}
\newcommand{\fpt}{\ensuremath{\operatorClassFPT}}

\DeclareMathOperator{\operatorClassFPA}{FPA}
\newcommand{\fpa}{\ensuremath{\operatorClassFPA}}

\DeclareMathOperator{\operatorInvariant}{Inv}
\newcommand{\Inv}{\ensuremath{\operatorInvariant}}

\DeclareMathOperator{\problemCSP}{CSP}
\newcommand{\csp}{\ensuremath{\problemCSP}}

\DeclareMathOperator{\problemDCSP}{\textsc{MinCSP}}
\newcommand{\dcsp}{\ensuremath{\problemDCSP}}
\DeclareMathOperator{\problemDCSPx}{\textsc{MinCSP}^{*}}
\newcommand{\dcspx}{\ensuremath{\problemDCSPx}}

\DeclareMathOperator{\evenRelation}{EVEN}
\newcommand{\even}{\ensuremath{\evenRelation}}

\DeclareMathOperator{\oddRelation}{ODD}
\newcommand{\odd}{\ensuremath{\oddRelation}}

\DeclareMathOperator{\naeRelation}{NAE^3}
\newcommand{\nae}{\ensuremath{\naeRelation}}

\DeclareMathOperator{\dupRelation}{DUP^3}
\newcommand{\dup}{\ensuremath{\dupRelation}}

\DeclareMathOperator{\problemAlmost}{\textsc{Almost 2-SAT}}
\newcommand{\almost}{\ensuremath{\problemAlmost}}

\DeclareMathOperator{\problemTwoSat}{\textsc{2-SAT}}
\newcommand{\twosat}{\ensuremath{\problemTwoSat}}

\DeclareMathOperator{\problemNearestCodeword}{\textsc{Nearest Codeword}}
\newcommand{\nearc}{\ensuremath{\problemNearestCodeword}}

\DeclareMathOperator{\galoisField}{GF(2)}
\newcommand{\gf}{\ensuremath{\galoisField}}

\DeclareMathOperator{\ewNotation}{\exists\wedge}
\newcommand{\ea}{\ensuremath{\ewNotation}}

\def\wp{0.98}

\def\odds{\textsc{Odd Set}\xspace}

\def\eoset{\textsc{Even/Odd Set}\xspace}

\newcommand{\la}{\langle}
\newcommand{\ra}{\rangle}

\newcommand{\cC}{\mathcal{C}}

\newcommand{\cS}{\mathcal{S}}

\newcommand{\cCx}{{\cC^*}}

\begin{document}

\maketitle

\begin{abstract}
  The minimum unsatisfiability version of a constraint satisfaction problem (\dcsp) asks for an assignment where the number of unsatisfied constraints is minimum possible, or equivalently, asks for a minimum-size set of constraints whose deletion makes the instance satisfiable.
For a finite set $\Gamma$ of constraints, we denote by \dcsp($\Gamma$) the restriction of the problem where each constraint is from $\Gamma$.
The polynomial-time solvability and the polynomial-time approximability of \dcsp($\Gamma$) were fully characterized by Khanna et al.~\cite{DBLP:journals/siamcomp/KhannaSTW00}.
Here we study the fixed-parameter (FP-) approximability of the problem: given an instance and an integer $k$, one has to find a solution of size at most $g(k)$ in time $f(k)\cdot n^{O(1)}$ if a solution of size at most $k$ exists.
We especially focus on the case of constant-factor FP-approximability.
We show the following dichotomy: for each finite constraint language $\Gamma$,
\begin{itemize}
\item either we exhibit a constant-factor FP-approximation for \dcsp($\Gamma$);
\item or we prove that \dcsp($\Gamma$) has no constant-factor FP-approximation unless $\textup{FPT}=\textup{W[1]}$.
\end{itemize}

In particular, we show that approximating the so-called \textsc{Nearest Codeword} within some constant factor is $\textup{W[1]}$-hard.
Recently, Arnab et al. \cite{Arnab18,Arnab18b} showed that such a $\textup{W[1]}$-hardness of approximation implies that \textsc{Even Set} is $\textup{W[1]}$-hard under randomized reductions.
Combining our results, we therefore settle the parameterized complexity of \textsc{Even Set}, a famous open question in the field. 
\end{abstract}

\section{Introduction}

Satisfiability problems and, more generally, Boolean constraint
satisfaction problems (CSPs) are basic algorithmic problems arising in
various theoretical and applied contexts. An instance of a Boolean CSP
consists of a set of Boolean variables and a set of constraints; each
constraint restricts the allowed combination of values that can appear
on a certain subset of variables. In the decision version of the
problem, the goal is to find an assignment that simultaneously
satisfies every constraint. One can also define optimization versions
of CSPs: the goal can be to find an assignment that maximizes the
number of satisfied constraints, minimizes the number of unsatisfied
constraints, maximizes/minimizes the weight (number of 1s) of the
assignment, etc. \cite{BooleanBook}.

Since these problems are usually NP-hard in their full generality, a
well-established line of research is to investigate how the complexity
of the problem changes for restricted versions of the problem. A large
body of research deals with language-based restrictions: given any
finite set $\Gamma$ of Boolean constraints, one can consider the
special case where each constraint is restricted to be a member of
$\Gamma$. The ultimate research goal of this approach is to prove a
{\em dichotomy theorem}: a complete classification result that
specifies for each finite constraint set $\Gamma$ whether the
restriction to $\Gamma$ yields an easy or hard problem.\footnote{Note that several authors have recently announced a proof of the dichotomy conjecture \cite{dichotomy_bulatov, dichotomy_rafiey, dichotomy_zhuk}.} Numerous
classification theorems of this form have been proved for various
decision and optimization versions for Boolean and non-Boolean CSPs
\cite{MR80d:68058,DBLP:journals/siamcomp/BulatovM14,DBLP:journals/jacm/Bulatov13,DBLP:journals/jcss/BulatovDGJJR12,DBLP:journals/tocl/Bulatov11,DBLP:journals/corr/abs-1010-0201,DBLP:journals/jacm/Bulatov06,DBLP:journals/jacm/DeinekoJKK08,DBLP:journals/siamcomp/JonssonKK06,DBLP:journals/jacm/KolmogorovZ13,DBLP:conf/stoc/ThapperZ13,DBLP:journals/cc/Marx05}.
In particular, for $\dcsp(\Gamma)$, which is the
optimization problem asking for an assignment minimizing the number of
unsatisfied constraints, Creignou et al.~\cite{BooleanBook} obtained a
classification of the polynomial-time approximability for every finite Boolean
constraint language $\Gamma$. The goal of this paper is to
characterize the approximability of Boolean $\dcsp(\Gamma)$ with
respect to the more relaxed notion of fixed-parameter approximability.

Parameterized complexity
\cite{DBLP:series/txcs/DowneyF13,MR2238686,DBLP:books/sp/CyganFKLMPPS15}
analyzes the running time of a computational problem not as a
univariate function of the input size $n$, but as a function of both
the input size $n$ and a relevant parameter $k$ of the input. For
example, given a \dcsp\ instance of size $n$ where we are looking for a
solution satisfying all but $k$ of the constraints, it is natural to
analyze the running time of the problem as a function of both $n$ and
$k$. We say that a problem with parameter $k$ is {\em fixed-parameter
  tractable (FPT)} if it can be solved in time $f(k)\cdot n^{O(1)}$
for some computable function $f$ depending only on $k$. Intuitively,
even if $f$ is, say, an exponential function, this means that problem
instances with ``small'' $k$ can be solved efficiently, as the
combinatorial explosion can be confined to the parameter $k$. This can
be contrasted with algorithms with running time of the form $n^{O(k)}$
that are highly inefficient even for small values of $k$. There are
hundreds of parameterized problems where brute force search gives trivial
$n^{O(k)}$ algorithms, but the problem can be shown to be FPT using
nontrivial techniques; see the recent textbooks by Downey and Fellows
\cite{DBLP:series/txcs/DowneyF13} and by Cygan et
al.~\cite{DBLP:books/sp/CyganFKLMPPS15}. In particular, there are
fixed-parameter tractability results and characterization theorems for
various CSPs
\cite{DBLP:journals/cc/Marx05,DBLP:journals/siamcomp/BulatovM14,DBLP:conf/mfcs/KratschMW10,DBLP:conf/icalp/KratschW10}.

The notion of fixed-parameter tractability has been combined with the
notion of approximability
\cite{papprox,DBLP:journals/eccc/ChenGG07,MR2466795,DBLP:journals/algorithmica/CaiH10,DBLP:conf/iwpec/ChitnisHK13}.
Following \cite{papprox,Marx08}, we
say that a minimization problem is {\em fixed-parameter approximable
  (FPA)} if there is an algorithm that, given an instance and an integer
$k$, in time $f_1(k)\cdot n^{O(1)}$ either returns a solution of cost at
most $f_2(k)\cdot k$ (where the function $f_2(k)\cdot k$ is non-decreasing), or correctly states that there is no solution of
cost at most $k$. The two crucial differences compared to the usual
setup of polynomial-time approximation is that (1) the running time is
not polynomial, but can have an arbitrary factor $f(k)$ depending only
on $k$ and (2) the approximation ratio is defined not as a function of the input size $n$ but as a function of $k$. In this paper,
we mostly focus on the case of constant-factor FPA, that is, when
$f_2(k)=c$ for some constant $c$.

Schaefer's Dichotomy Theorem \cite{MR80d:68058} identified six classes
of finite Boolean constraint languages (0-valid, 1-valid, Horn,
dual-Horn, bijunctive, affine) for which the decision CSP is
polynomial-time solvable, and shows that every language $\Gamma$
outside these classes yields NP-hard problems. Therefore, one has to study \dcsp\ only within these six classes, as it
is otherwise already NP-hard to decide if the optimum is $0$ or not, making
approximation or fixed-parameter tractability irrelevant. Within these
classes, polynomial-time approximability and fixed-parameter
tractability seem to appear in orthogonal ways: the classes where we
have positive results for one approach is very different from the
classes where the other approach helps. For example, \textsc{2SAT Deletion} (also called \textsc{Almost 2SAT}) is fixed-parameter
tractable \cite{almost2satfpt,DBLP:journals/talg/LokshtanovNRRS14},
but has no polynomial-time approximation algorithm with constant
approximation ratio, assuming the Unique Games Conjecture
\cite{DBLP:journals/cc/ChawlaKKRS06}. On the other hand, if $\Gamma$
consists of the three constraints $(x)$, $(\bar x)$, and
$(a\to b)\wedge (c\to d)$, then the problem is W[1]-hard
\cite{DBLP:journals/ipl/MarxR09}, but belongs to the class
IHS-B and
hence admits a constant-factor approximation in polynomial time
\cite{DBLP:journals/siamcomp/KhannaSTW00}.\footnote{IHS-B stands for Implicative Hitting Set-Bounded, see
  definition in Section~\ref{sec:preliminaries}.}

By investigating constant-factor FP-approximation, we are identifying
a class of tractable constraints that unifies and generalizes the
polynomial-time constant-factor approximable and fixed-parameter
tractable cases. We observe that if each constraint in $\Gamma$ can be
expressed by a 2SAT formula (i.e., $\Gamma$ is bijunctive), then we
can treat the \dcsp\ instance as an instance of \textsc{2SAT
  Deletion}, at the cost of a constant-factor loss in the
approximation ratio. Thus the fixed-parameter tractability of
\textsc{2SAT Deletion} implies \dcsp\ has a constant-factor
FP-approximation if the finite set $\Gamma$ is bijunctive.  If
$\Gamma$ is in IHS-B, then \dcsp\ is known to
have a constant-factor approximation in polynomial time, which clearly gives another
class of constant-factor FP-approximable constraints.
Our main results show that these two classes cover all the easy cases with respect to FP-approximation (see Section~\ref{sec:preliminaries} for the definitions involving properties of constraints) unless $\textup{FPT}=\textup{W[1]}$.
\begin{theorem}\label{th:mainintro}
Let $\Gamma$ be a finite Boolean constraint language.
\begin{enumerate}
\item If $\Gamma$ is bijunctive or IHS-B, then $\dcsp(\Gamma)$ has a constant-factor FP-approximation.
\item Otherwise, $\dcsp(\Gamma)$ has no constant-factor FP-approximation, unless $\textup{FPT}=\textup{W[1]}$.
\end{enumerate}
\end{theorem}

Moreover, in the second case (when $\Gamma$ is neither bijunctive nor IHS-B), if $\Gamma$ is also \emph{not} affine, we can show a stronger inapproximability result; namely, that $\dcsp(\Gamma)$ has no FP-approximation for any function of the optimum value, unless $\textup{FPT}=\textup{W[P]}$.
Note that this result is stronger in two different ways: it rules out not only constant-factor but any ratio of approximation, and it relies on a weaker assumption.

Given a linear code over $\mathbb{F}_2$ and a vector, the \textsc{Nearest Codeword} (\textsc{NC}) problem asks for a codeword in the code that has minimum Hamming distance to the given vector.
There are various equivalent formulations of this  problem: \textsc{Odd Set} is a variant of \textsc{Hitting Set} where one has to select at most $k$ elements to hit each set an odd number of times, and it is also possible to express the problem as finding a solution to a system of linear equations over $\mathbb{F}_2$ that minimizes the number of unsatisfied equations.
Dinur et al. \cite{Dinur03} showed that approximating \textsc{Nearest Codeword} within ratio $n^{1/O(\log \log n)}$ is $\textup{NP}$-hard.
However, this does not give any evidence against constant-factor FP-approximation. Building on the work of Lin \cite{DBLP:conf/soda/Lin15} proving hardness for \textsc{Biclique} and related problems, we are able to show that even polylogarithmic FP-approximation is unlikely for \textsc{Odd Set}.

\begin{theorem}\label{th:main-odd-set}
\odds has no ratio $\log^{O(1)} k$ FP-approximation, unless $\fpt=\W 1$.
\end{theorem}

This theorem is the most technically involved part of the paper, as well as the most interesting contribution.
Furthermore, Arnab et al. \cite{Arnab18,Arnab18b} showed that if it is $\W 1$-hard to approximate \textsc{Nearest Codeword}/\textsc{Odd Set} within some constant factor, then that would give a randomized $\W 1$-hardness construction for \textsc{Even Set}.
By combining their result with Theorem~\ref{th:main-odd-set}, we obtain:
\begin{theorem}\label{th:even-set}
  \textsc{Even Set} is $\W 1$-hard under randomized reductions.
\end{theorem}
This settles a well-known open question in parameterized complexity.

Post's lattice is a very useful tool for classifying the complexity of
Boolean CSPs (see e.g., \cite{refining,weightedSAT,testingCSPs}). A
(possibly infinite) set $\Gamma$ of constraints is a co-clone if it is
closed under pp-definitions, that is, whenever a relation $R$ can be
expressed by relations in $\Gamma$ using only equality, conjunctions,
and projections, then relation $R$ is already in $\Gamma$. Post's
co-clone lattice characterizes every possible co-clone of Boolean
constraints. From the complexity-theoretic point of
view, Post's lattice becomes very relevant if the complexity of the
CSP problem under study does not change by adding new pp-definable
relations to the set $\Gamma$ of allowed relations. For example, this
is true for the decision version of Boolean CSP. In this case, it is
sufficient to determine the complexity for each co-clone in the
lattice, and a complete classification for every finite set $\Gamma$
of constraints follows. For \dcsp, neither the polynomial-time
solvability nor the fixed-parameter tractability of the problem is
closed under pp-definitions, hence Post's lattice cannot be used
directly to obtain a complexity classification. However, as observed by
Khanna et al.~\cite{DBLP:journals/siamcomp/KhannaSTW00} and
subsequently exploited by Dalmau et
al.~\cite{DBLP:journals/toct/DalmauK13,DBLP:conf/soda/DalmauKM15}, the
constant-factor approximability of \dcsp\ is closed under
pp-definitions (modulo a small technicality related to equality
constraints). We observe that the same holds for constant-factor
FP-approximability and hence Post's lattice can be used for our
purposes. Thus, the classification result amounts to identifying
the maximal easy and the minimal hard co-clones.

The paper is organized as follows.
Sections~\ref{sec:preliminaries} and \ref{sec:posts-lattice-co} contain preliminaries on CSPs, approximability, Post's lattice, and reductions.
A more technical restatement of Theorem~\ref{th:mainintro} in terms of co-clones is stated at the end of Section~\ref{sec:posts-lattice-co}.
Section \ref{sec:csps-with-fpa} gives FPA algorithms, Section \ref{affine_stuff} establishes the equivalence of some CSPs with \odds, Section~\ref{sec:odd-set-probably} shows the hardness result for \odds (Theorem~\ref{th:main-odd-set}), and Section~\ref{sec:horn-dual-horn} proves inapproximability results for the remaining boolean \textsc{MinCSP}s.

\section{Preliminaries}
\label{sec:preliminaries}
\textbf{Constraint Satisfaction Problems (CSPs).} A subset $R$ of $\{0,1\}^n$ is called an $n$-ary \emph{Boolean relation}. If $n = 2$, relation $R$ is \emph{binary}. In this paper, a \emph{constraint language} $\Gamma$ is a finite collection of finitary Boolean relations. When a constraint language $\Gamma$ contains only a single relation $R$,  i.e., $\Gamma = \{R\}$, we write $R$ instead of $\{R\}$. The decision version of CSP, restricted to finite constraint language $\Gamma$ is defined as:

\begin{center}
	\noindent\framebox{\begin{minipage}{\wp\textwidth}
			$\csp(\Gamma)$\\
			\emph{Input:} A pair $\langle V,\cC \rangle$, where \begin{itemize}
				\item $V$ is a set of variables,
				\item $\cC$ is a multiset of constraints $\{C_1,\dots,C_q\}$, i.e., $C_i = \langle s_i,R_i \rangle$, where $s_i$ is a tuple of variables of length $n_i$, and $R_i \in \Gamma$ is an $n_i$-ary relation.
			\end{itemize}
			\emph{Question:} Does there exist a solution, that is, a function $\varphi : V \rightarrow \{0,1\}$ such that for each constraint $\langle s,R \rangle \in \cC$, with $s = \langle v_1,\dots,v_n \rangle$, the tuple $\langle\varphi(v_1),\dots,\varphi(v_n)\rangle$ belongs to $R$?
		\end{minipage}}
	\end{center}

Note that we can alternatively look at a constraint as a Boolean function $f:\{0,1\}^n \rightarrow \{0,1\}$, where $n$ is a non-negative integer called the arity of $f$. We say that $f$ is satisfied by an assignment $s \in \{0,1\}^n$ if $f(s) = 1$. For example, if $f(x,y) = x + y \mod 2$, then the corresponding relation is $\{(0,1),(1,0)\}$; we also denote addition modulo $2$ with $x \oplus y$.

We recall the definition of a few well-known classes of constraint languages. A Boolean constraint language $\Gamma$ is:
\begin{itemize}
\item \textit{0-valid} (resp.~\textit{1-valid}), if each $R \in \Gamma$
  contains a tuple in which all entries are $0$ (resp.~$1$);
\item \textit{k-IHS-B+} (resp.~\textit{k-IHS-B--}), where
  $k \in \mathbb{Z}^+$, if each $R \in \Gamma$ can be expressed by a
  conjunction of clauses of the form $\bar x$, $\bar x \vee y$, or
  $x_1 \vee \dots \vee x_k$ (resp.~$x$, $\bar x \vee y$, or
  $\bar x_1 \vee \dots \vee \bar x_k$); \emph{IHS-B+} (resp.~\emph{IHS-B--})
  stands for $k$-IHS-B+ (resp.~$k$-IHS-B--) for some $k$; \textit{IHS-B}
  stands for IHS-B+ or IHS-B--;
  \item \textit{bijunctive}, if each $R \in \Gamma$ can be expressed
    by a conjunction of binary clauses;
  \item \textit{Horn} (\textit{dual-Horn}), if each $R \in \Gamma$ can
    be expressed by a conjunction of Horn (dual-Horn) clauses, i.e.,
    clauses that have at most one positive (negative) literal;
  \item \textit{affine}, if each relation $R \in \Gamma$ can be
    expressed by a conjunction of relations defined by equations of
    the form $x_1 \oplus \dots \oplus x_n = c$, where $c \in \{0,1\}$;
  \item \textit{self-dual} if for each relation $R \in \Gamma$,
    $(a_1,\dots,a_n) \in R \Rightarrow (\bar a_1,\dots, \bar a_n) \in
    R$.
  \end{itemize}

\noindent\framebox{\begin{minipage}{\wp\textwidth}
$\dcsp(\Gamma)$ \\
\emph{Input:} An instance $\langle V,\cC \rangle$ of $\csp(\Gamma)$, and an integer $k$.\\
\emph{Question:} Is there a deletion set $W \subseteq \cC$ such that $|W| \leq k$, and the $\csp(\Gamma)$-instance $\langle V, \cC \setminus W \rangle$ has a solution?
\end{minipage}}
\medskip

\noindent\framebox{\begin{minipage}{\wp\textwidth}
    $\dcspx(\Gamma)$\\
\emph{Input:} An instance $\langle V,\cC \rangle$ of $\csp(\Gamma)$, a subset $\cCx\subseteq \cC$ of undeletable constraints, and an integer $k$.
\\
\emph{Question:} Is there a deletion set $W \subseteq \cC\setminus \cCx$ such that $|W| \leq k$ and the $\csp(\Gamma)$-instance $\langle V, \cC \setminus W \rangle$ has a solution?
\end{minipage}}
\medskip

For every finite constraint language $\Gamma$, we consider the problem \dcsp\ above. For technical reasons, it will be convenient to work with a slight generalization of the problem, \dcspx (defined above), where we can specify that certain constraints are ``undeletable.''
For these two problems, a set of potentially more than $k$ constraints whose removal yields a satisfiable instance is called a \emph{feasible solution}.
Note that, contrary to \dcsp\ for which removing all the constraints constitute a trivially feasible solution, it is possible that an instance of \dcspx\ has no feasible solution.
A \emph{feasible instance} is an instance that admits at least one feasible solution.

\textbf{Reductions.}
We will use two types of reductions to connect the approximability of optimization problems.
The first type perfectly preserves the optimum value (or cost) of instances.

\begin{definition}
  An optimization problem $A$ has a {\em cost-preserving reduction} to problem $B$ if there are two polynomial-time computable functions $F$ and $G$ such that
\begin{enumerate}
\item For any feasible instance $I$ of $A$, $F(I)$ is a feasible instance of $B$ having the same optimum cost as $I$.
\item For any feasible instance $I$ of $A$, if $S'$ is a feasible solution for $F(I)$, then $G(I,S')$ is a feasible solution of $I$ having cost at most the cost of $F(I)$.
\end{enumerate}

\end{definition}
The following easy lemma shows that the existence of undeletable constraints does not make the problem significantly more general.
Note that, in the previous definition, if instance $I$ has no feasible solution, then the behavior of $F$ on $I$ is not defined.
\begin{lemma}\label{lem:reducerestricted}
There is a cost-preserving reduction from \dcspx\ to \dcsp.
\end{lemma}

\begin{proof}
	The function $F$ on a feasible instance $I$ of \dcspx\ is defined the
	following way. Let $m$ be the number of constraints. We construct
	$F(I)$ by replacing each undeletable constraint with $m+1$
	copies. If $I$ is a feasible instance of \dcspx, then $I$ has a
	solution with at most $m$ deletions, which gives a solution of
	$F(I)$ as well, showing that $OPT(F(I))\le OPT(I)\le m$. Conversely,
	$OPT(F(I))\le m$ implies that an optimum solution of $F(I)$ uses
	only the deletable constraints of $I$, otherwise it would need to
	delete all $m+1$ copies of an undeletable constraints. Thus
	$OPT(I)\le OPT(F(I))$ and hence $OPT(I)=OPT(F(I))$ follows.
	
	The function $G(I,S')$ on a feasible instance $I$ of \dcspx\ and a feasbile solution $S'$ of $F(I)$ is defined the following way. If $S'$ deletes only the deletable constraints of $I$, then $G(I,S')=S'$ is also a feasible solution of $I$ with the same cost. Otherwise, if $S'$ deletes at least one undeletable constraint, then it has cost at least $m+1$, as it has to delete all $m+1$ copies of the constraint. Now we define $G(I,S')$ to be the set of all (at most $m$) deletable constraints; by assumption, $I$ is a feasible instance of \dcspx, hence $G(I,S')$ is a feasbile solution of cost at most $m+1$
\end{proof}

The second type of reduction that we use is the standard notion of A-reductions \cite{DBLP:journals/iandc/CrescenziP91},  which preserve approximation ratios up to constant factors. We slightly deviate from the standard definition by not requiring any specific behavior of $F$ when $I$ has no feasible solution.

\begin{definition}
  A minimization problem $A$ is {\em A-reducible} to problem $B$ if there are two polynomial-time computable functions $F$ and $G$ and a constant $\alpha$ such that
\begin{enumerate}
\item For any feasible instance $I$ of $A$, $F(I)$ is a feasible instance of $B$.
\item For any feasible instance $I$ of $A$, and any feasible solution $S'$ of $F(I)$, $G(I,S')$ is a feasible solution for $I$.
\item For any feasible instance $I$ of $A$, and any $r\ge 1$, if $S'$ is an $r$-approximate feasible solution for $F(I)$, then $G(I,S')$ is an $(\alpha r)$-approximate feasible solution for $I$.
\end{enumerate}

\end{definition}

\begin{proposition}
If optimization problem $A$ is A-reducible to optimization problem $B$ and $B$ admits a constant-factor FPA algorithm, then $A$ also has a constant-factor FPA algorithm.
\end{proposition}

\section{Post's lattice, co-clone lattice, and a simple reduction}
\label{sec:posts-lattice-co}

A \textit{clone} is a set of Boolean functions that contains all projections (that is, the functions $f(a_1,\dots,a_n) = a_k$ for $1 \leq k \leq n$) and is closed under arbitrary composition. \textit{All} clones of Boolean functions were identified by Post \cite{Post41}, and he also described their inclusion structure, hence the name Post's lattice. To make use of this lattice for CSPs, Post's lattice can be transformed to another lattice whose elements are not sets of functions closed under composition, but sets of relations closed under the following notion of definability.
\begin{definition}\label{pp_def}
	Let $\Gamma$ be a constraint language over some domain $A$.	We say that a relation $R$ is \emph{pp-definable} from $\Gamma$ if there exists a (primitive positive) formula
	$\varphi(x_1,\dots,x_k) \equiv \exists y_1,\dots,y_\ell \psi(x_1,\dots,x_k,y_1,\dots,y_\ell)$,
	where $\psi$ is a conjunction of atomic formulas with relations in $\Gamma$ and $EQ_A$ (the binary relation $\{(a,a) : a \in A\}$) such that for every $(a_1,\dots,a_k) \in A^k$ $(a_1,\dots,a_k) \in R \text{ if and only if } \varphi(a_1,\dots,a_k)$
	holds. If $\psi$ does not contain $EQ_A$, then we say that $R$ is \emph{pp-definable} from $\Gamma$ \emph{without equality}. For brevity, we often write ``\ea-definable'' instead of ``pp-definable without equality''. If $S$ is a set of relations, $S$ is \emph{pp-definable} (resp.~\ea-definable) from $\Gamma$ if every relation in $S$ is $pp$-definable (resp.~\ea-definable) from $\Gamma$.
\end{definition}
For a set of relations $\Gamma$, we denote by $\la \Gamma \ra$ the set of all relations that can be pp-defined over $\Gamma$. We refer to $\la \Gamma \ra$ as the \textit{co-clone} generated by $\Gamma$. The set of all co-clones forms a lattice. To give an idea about the connection between Post's lattice and the co-clone lattice, we briefly mention the following theorem, and refer the reader to, for example, \cite{bases} for more information. Roughly speaking, the following theorem says that the co-clone lattice is essentially Post's lattice turned upside down, i.e., the inclusion between neighboring nodes are inverted.
\begin{theorem}[\cite{poschel}, Theorem 3.1.3]
	The lattices of Boolean clones and Boolean co-clones are anti-isomorphic.
\end{theorem}

Using the above comments, it can be seen (and it is well known) that the lattice of Boolean co-clones has the structure shown in Figure~\ref{classfigure}.\footnote{We thank Heribert Vollmer and Yuichi Yoshida for giving us access to their Post's lattice diagrams.}
In the figure, if co-clone $C_2$ is above co-clone $C_1$, then $C_2 \supset C_1$. The names of the co-clones are indicated in the nodes\footnote{If the name of a clone is $\text{L}_3$, for example, then the corresponding co-clone is $\Inv(\text{L}_3)$ (\Inv\ is defined, for example, in \cite{bases}), which is denoted by $\mathrm{IL_3}$.}, where we follow the notation of B\"ohler et al~\cite{bases}.

\begin{figure}[!t]
	\centering{
		\includegraphics[scale=0.55]{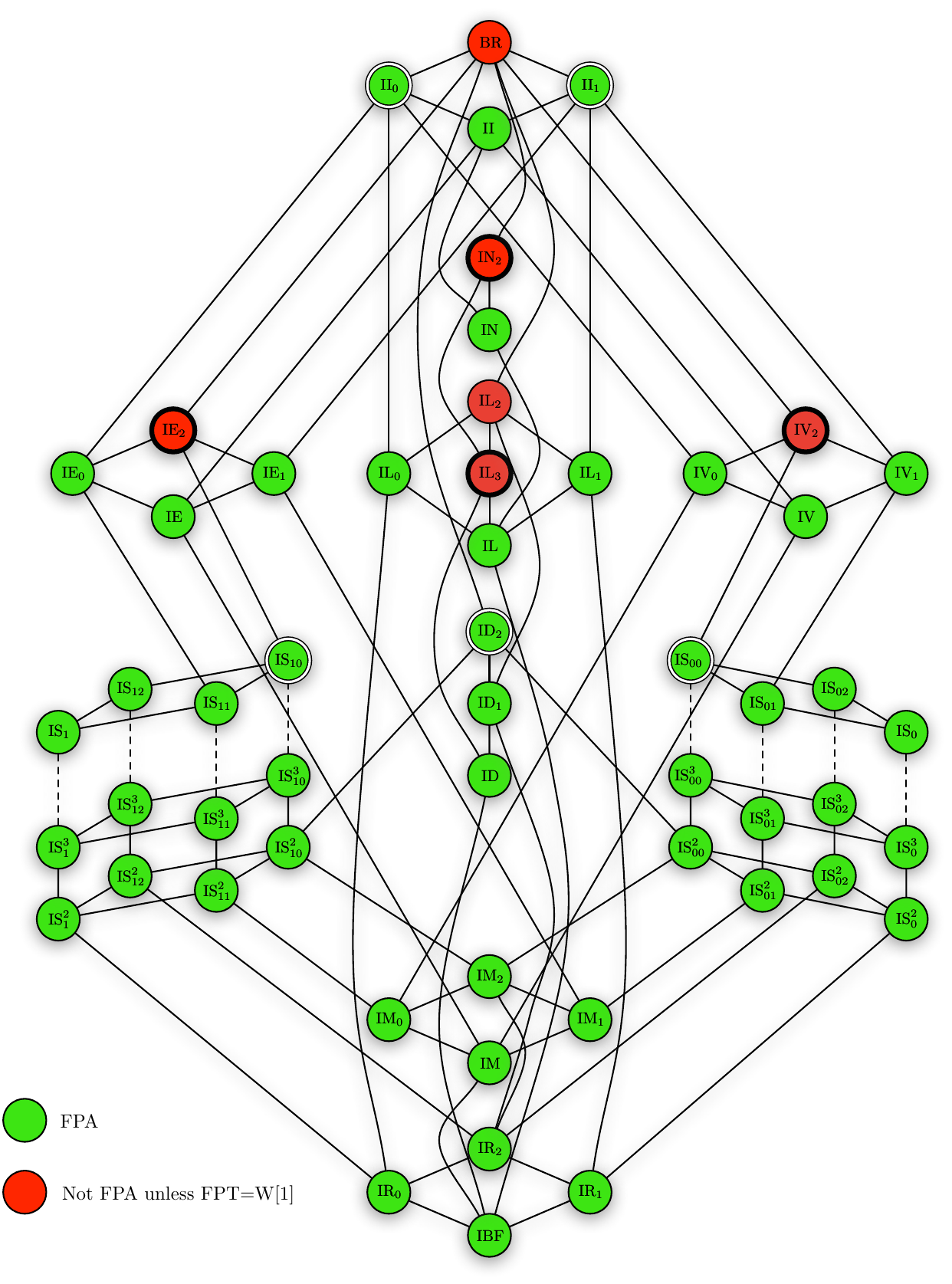}}
	\caption{Classification of Boolean CSPs according to constant ratio fixed-parameter approximability. The results for the few co-clones represented with an extra circle (in black for negative results and in white for positive results) imply the full classification.}
	\label{classfigure}
\end{figure}

For a co-clone $C$ we say that a set of relations $\Gamma$ is a \textit{base} for $C$ if $C = \la \Gamma \ra$, that is, any relation in $C$ can be pp-defined using relations in $\Gamma$. B\"{o}hler et al.\ give bases for all co-clones in \cite{bases}, and the reader can consult this paper for details. We reproduce this list in Table~\ref{cc_table}.\footnote{We note that $\even^4$ can be pp-defined using $\dup$. Therefore the base  $\{\dup, \even^4, x \oplus y\}$ given by B\"ohler et al.~\cite{bases} for $\mathrm{IN_2}$ can be actually simplified to $\{\dup, x \oplus y\}$.}

\begin{table}[!t]
  \centering
  \resizebox{\textwidth}{!}{
  \begin{tabular}{lll|lll}
    \hline
    \textbf{Co-clone}    & \textbf{Order}    & \textbf{Base} 	& \textbf{Co-clone}    & \textbf{Order}    &
                                                                   \textbf{Base}
    \\ \hline
    IBF         & 0        & \{=\}, $\{\emptyset\}$ & IS$_{10}$   & $\infty$ & $\{\text{NAND}^m | m \geq 2\} \cup \{x, \bar{x}, x \rightarrow y\}$ \\
    \hline
    IR$_0$      & 1        & $\{\bar{x}\}$  & ID & 2 & $\{x
                                                       \oplus
                                                       y\}$ \\

    \hline
    IR$_1$      & 1        & $\{x\}$ & 		ID$_1$      &
                                                              2
                                                                 &
                                                                   $\{x
                                                                   \oplus
                                                                   y,
                                                                   x\}$,
                                                                   every
                                                                   $R
                                                                   \in
                                                                   \{\{(a_1,a_2,a_3),$\\
                &&&&& \quad                       $(b_1,b_2,b_3)\}|\exists c \in \{1,2\} \text{ such
                      that }$      \\
                &&&&& \quad $\sum_{i=1}^3 a_i = \sum_{i=i}^3
                      b_i = c\}$                              \\ \hline
    IR$_2$      & 1        & $\{x,\bar{x}\}, \{x \bar{x}\}$ &    ID$_2$      &       2   & $\{x \oplus y, x \rightarrow y\}, \{x \bar{y}, \bar{x}yz\}$                                                                \\ \hline
    IM          & 2        & $\{x \rightarrow y\}$    &       IL          &          4& $\{\even^4\}$                                                                                                              \\
    \hline
    IM$_1$      & 2        & $\{x \rightarrow y, x\}, \{x \wedge
                             (y \rightarrow z)\}$ &           IL$_0$      &         3 & $\{\even^4, \bar{x}\}, \{\even^3\}$                                                                                                                                                         \\
    \hline
    IM$_0$      & 2        & $\{x \rightarrow y, \bar{x}\},
                             \{\bar{x} \wedge (y \rightarrow
                             z)\}$ &           IL$_1$      &
                                                             3 &
                                                                 $\{\even^4, x\}, \{\odd^3\}$                                                                                                                                                      \\ \hline
    IM$_2$      & 2        & $\{x \rightarrow y, x, \bar{x}\},
                             \{x \rightarrow y, \overline{x
                             \rightarrow y}\},$  & IL$_2$
                                                      &
                                                        3 &
                                                            $\{\even^4,
                                                            x,\bar{x}\},
                                                            \text{every
                                                            }
                                                            \{\even^n,x\}$
    \\
                && $\{x \bar{y}
                   \wedge (u \rightarrow v)\}$ &&&$\text{\quad where } n \geq 3 \text{ is odd}$                             \\  \hline
    IS$_0^m$    & m        & $\{\text{OR}^m\}$     &  IL$_3$      &        4  & $\{\even^4, x \oplus y\}, \{\odd^4\}$                                                                                      \\  \hline

    IS$_1^m$    & m        & $\{\text{NAND}^m\}$         &      IV          &          3& $\{x \vee y \vee \bar{z}\}$                                                                                                \\   \hline
    IS$_0$      & $\infty$ & $\{\text{OR}^m | m \geq 2\}$    &
                                                               IV$_0$
                                                      &
                                                        3 &
                                                            $\{x \vee y \vee \bar{z},\bar{x}\}$ \\ \hline
    IS$_1$      & $\infty$ & $\{\text{NAND}^m | m \geq 2\}$  & IV$_1$      &         3 & $\{x \vee y \vee \bar{z},x\}$          \\       \hline
    IS$_{02}^m$ & m        & $\{\text{OR}^m, x, \bar{x}\}$ &
                                                             IV$_2$      &        3  & $\{x \vee y \vee \bar{z},x,\bar{x}\}$      \\ \hline
    IS$_{02}$   & $\infty$ & $\{\text{OR}^m | m \geq 2\} \cup
                             \{x, \bar{x}\}$        &           IE          &          3& $\{\bar{x} \vee \bar{y} \vee z\}$                                                                                             \\ \hline
    IS$_{01}^m$ & m        & $\{\text{OR}^m, x \rightarrow y\}$     &     IE$_1$      &          3& $\{\bar{x} \vee \bar{y} \vee z,x\}$                                                                           \\  \hline
    IS$_{01}$   & $\infty$ & $\{\text{OR}^m | m \geq 2\} \cup
                             \{x \rightarrow y\}$    &
                                                       IE$_0$      &
                                                                     3
                                                                 &
                                                                   $\{\bar{x} \vee \bar{y} \vee z,\bar{x}\}$             \\  \hline
    IS$_{00}^m$ & m        & $\{\text{OR}^m, x, \bar{x}, x \rightarrow
                             y\}$      &       IE$_2$      &      3
                                                                 &
                                                                   $\{\bar{x} \vee \bar{y} \vee z,x,\bar{x}\}$                                                                                                                                                         \\  \hline
    IS$_{00}$   & $\infty$ & $\{\text{OR}^m | m \geq 2\} \cup \{x, \bar{x}, x \rightarrow y\}$      &      IN          & 3        & $\{\dup\}$                                                                                                                                                                        \\ \hline
    IS$_{12}^m$ & m        & $\{\text{NAND}^m, x, \bar{x}\}$          &          IN$_2$      & 3        & $\{\dup, x \oplus y\}, \{\nae\}$                                                                                                                                                                             \\ \hline
    IS$_{12}$   & $\infty$ & $\{\text{NAND}^m | m \geq 2\} \cup \{x, \bar{x}\}$      &         II          &       3   & $\{\even^4, x \rightarrow y\}$                                                                           \\ \hline
    IS$_{11}^m$ & m        & $\{\text{NAND}^m, x \rightarrow y\}$        &             II$_0$      & 3        & $\{\even^4, x \rightarrow y, \bar{x}\}, \{\dup,x \rightarrow y\}$                                                    \\ \hline
    IS$_{11}$   & $\infty$ & $\{\text{NAND}^m | m \geq 2\} \cup
                             \{x \rightarrow y\}$           &           II$_1$      & 3        & $\{\even^4, x \rightarrow y, x\}, \{x \vee (x \oplus z)\}$                                                                                                                       \\ \hline
    IS$_{10}^m$ & m        & $\{\text{NAND}^m, x, \bar{x}, x \rightarrow
                             y\}$      &        BR          & 3        &
                                                                         $\{\even^4, x \rightarrow y, x, \bar{x}\},$\\
                &&&&&  $\text{\quad}\{\text{1-IN-3}\}, \{x \vee (x \oplus z)\}$                                            \\ \hline
  \end{tabular}}
  \caption{Bases for all Boolean co-clones. (See \cite{bases} for a complete definition of relations that appear.) The order of a co-clone is the minimum over all bases of the maximum arity of a relation in the base. The order is defined to be infinite if there is no finite base for that co-clone.}
  \label{cc_table}
\end{table}

It is well-known that pp-definitions preserve the complexity of the decision version of CSP: if $\Gamma_2\subseteq \la \Gamma_1\ra$ for two finite languages $\Gamma_1$ and $\Gamma_2$, then there is a natural polynomial-time reduction from $\csp(\Gamma_2)$ to $\csp(\Gamma_1)$. The same is not true for \dcsp: the approximation ratio can change in the reduction. However, it has been observed that this change of the approximation ratio is at most a constant (depending on $\Gamma_1$ and $\Gamma_2$) \cite{DBLP:journals/siamcomp/KhannaSTW00,DBLP:journals/toct/DalmauK13,DBLP:conf/soda/DalmauKM15}; we show the same here in the context of parameterized reductions.

\begin{lemma}\label{pp_reduction}
Let $\Gamma$ be a constraint language, and $R$ be a relation that is pp-definable over $\Gamma$ without equality.  Then there is an A-reduction from $\dcsp(\Gamma \cup \{R\})$ to $\dcsp(\Gamma)$.
\end{lemma}

\begin{proof}
 Let $I$ be an instance of $\dcsp(\Gamma \cup \{R\})$. Let $\varphi(x_1,\dots,x_k)$ be a primitive positive formula defining $R$ from $\Gamma$. Then $\varphi$ is of the form $\exists y_1,\dots,y_\ell \psi(x_1,\dots,$ $x_k,y_1,\dots,y_\ell)$, where $\psi$ is the quantifier-free part of $\varphi$. The key and well-known (in similar contexts) observation is that $\psi$ can be alternatively seen as an instance of $\dcsp(\Gamma)$. More precisely, we define the instance associated to $\psi$, $I_\psi$, as the instance that has variables $x_1, \dots, x_k, y_1, \dots, y_\ell$ and contains for every atomic formula $S(v_1,\dots,v_r)$ in $\psi$, the constraint $\la(v_1,\dots,v_r), S \ra$. It follows that for any assignment $s : {x_1,\dots,x_k,y_1,\dots,y_\ell} \rightarrow A$, $s$ is a solution of $I_\psi$ if and only if $\psi(s(x_1),\dots, s(x_k), s(y_1),\dots, s(y_\ell))$ holds.
	
	We obtain an instance $I'$ of $\dcsp(\Gamma)$ from $I$ using the following replacement. For each constraint $C = \la(u_1,\dots,u_q), R\ra$ in $I$, we identify the quantifier-free part $\psi_C(u_1,\dots,u_q,y_{C,1},$ $\dots,y_{C,l})$ of the formula corresponding to $C$, and then replace $C$ with the set of constraints of the instance $I_{\psi_C}$, where $y_{C,1},\dots,y_{C,l}$ are newly introduced variables. We leave the rest of the constraints intact.
	
	Any deletion set $X_I$ for $I$ is translated to a deletion set $X_{I'}$ of $I'$ as follows. If $C = \la (u_1,\dots,u_q), P \ra \in X_I$ and $P \neq R$, then we place $\la (u_1,\dots,u_q), P \ra \in X_{I'}$. If $P = R$, then we place all the constraints that replaced $C$ into $X_{I'}$. Since the number of these constraints is bounded by a constant, we obtain only a constant blow-up in the solution size. The converse can be shown similarly.
\end{proof}

By repeated applications of Lemma~\ref{pp_reduction}, the following corollary establishes that we need to provide approximation algorithms only for a few \dcsp{}s, and these algorithms can be used for other \dcsp{}s associated with the same co-clone.
\begin{corollary}\label{FPA_spreads}
  Let $C$ be a co-clone and $B$ be a base for $C$. If the equality
  relation can be \ea-defined from $B$, then for any finite
  $\Gamma \subseteq C$, there is an A-reduction from $\dcsp(\Gamma)$ to $\dcsp(B)$.
\end{corollary}

For hardness results, we wish to argue that if a co-clone $C$ is hard, then any constraint language $\Gamma$ generating the co-clone is hard. However, there are two technical issues. First, co-clones are infinite and our constraint languages are finite. Therefore, we formulate this requirement instead by saying that a finite base $B$ of the co-clone $C$ is hard. Second, pp-definitions require equality relations, which may not be expressible by $\Gamma$. However, as the following theorem shows, this is an issue only if $B$ contains relations where the coordinates are always equal (which will not be the case in our proofs).
A $k$-ary relation $R$ is \emph{irredundant} if for every two different coordinates $1 \leq i < j \leq k$, $R$
contains a tuple $(a_1,\dots,a_k)$ with $a_i \neq a_j$. A set of relations $S$ is \emph{irredundant} if any relation in $S$ is irredundant.
\begin{theorem}[\cite{Geiger68,Bodnarchuk69}]\label{irredundant}
	If $S \subseteq \la \Gamma \ra$ and $S$ is irredundant, then $S$ is $\exists \wedge$-definable from $\Gamma$.
\end{theorem}
Thus, considering an irredundant base $B$ of co-clone $C$, we can formulate the following result.
\begin{corollary}\label{cor:irredbase}
	Let $B$ be an irredundant base for some co-clone $C$. If  $\Gamma$ is a finite constraint language with $C \subseteq \la \Gamma \ra$, then there is an A-reduction from $\dcsp(B)$ to $\dcsp(\Gamma)$.
\end{corollary}

\begin{proof}
	Since $B \subseteq \la \Gamma \ra$ and $B$ is irredundant, $B$ can be $\ea$-defined over $\Gamma$ using Theorem~\ref{irredundant}.
	Then repeated applications of Lemma~\ref{pp_reduction} shows the existence of the reduction from $\dcsp(B)$ to $\dcsp(\Gamma)$.
\end{proof}

By the following lemma, if the constraint language is self-dual, then we can assume that it also contains the constant relations.
\begin{lemma}\label{adding_constants}
	Let $\Gamma$ be a self-dual constraint language. Assume that $x \oplus y \in \Gamma$. Then there is a cost-preserving reduction from $\dcsp(\Gamma \cup \{x,\bar{x}\})$ to $\dcsp(\Gamma)$.
\end{lemma}

\begin{proof}
	Let $I$ be an instance of $\dcsp(\Gamma \cup \{x,\bar{x}\})$. We construct an instance $J$ of $\dcspx(\Gamma)$ such that a deletion set of size $k$ of $I$ corresponds to a deletion set of size $k$ of $J$; then the result for $\dcsp$ follows from the fact that there is a cost-preserving reduction from \dcspx\ to \dcsp (Lemma~\ref{lem:reducerestricted}). Every constraint of the form $R(x_1,\dots,x_r)$ in $I$ is placed into $J$. We introduce an undeletable constraint $x \oplus y$ to $J$, where $x$ and $y$ are new variables. For any constraint $v = 0$ in $I$, we add a constraint $v \oplus y$, and for any constraint $v = 1$ in $I$, we add a constraint $v \oplus x$ in $J$.
	
	Let $W_I$ be a deletion set for $I$. To obtain a deletion set $W_J$ of the same size for $J$,  any constraint that is not of the form $v=0$ or $v=1$ is placed into $W_J$. For any constraint $v=0 \in W_I$, we place the constraint $v \oplus y$ into $W_J$, and for any constraint $v=1 \in W_I$, we place the constraint $v \oplus x$ into $W_J$. Then assigning $0$ to $x$ and $1$ to $y$, and for the remaining variables of $J$ using the assignment for the variables of $I$, we obtain a satisfying assignment for $J$.
	
	The converse can be done by essentially reversing the argument, except that we might need to use the complement of the satisfying assignment for $J$ to satisfy constraints of the form $v = 0$ and $v = 1$.
\end{proof}

The following theorem states our classification in terms of co-clones.
\begin{theorem}\label{trichotomy}
  Let $\Gamma$ be a finite set of Boolean relations.
  \begin{enumerate}
    \item If $\la \Gamma \ra \subseteq C$ (equivalently, if $\Gamma \subseteq C$), with $C \in \{\mathrm{II_0}, \mathrm{II_1}, \mathrm{IS_{00}}, \mathrm{IS_{10}}, \mathrm{ID_2}\}$, then $\dcsp(\Gamma)$ has a constant-factor \fpa\ algorithm. (Note in these cases $\Gamma$ is $0$-valid, $1$-valid, IHS-B+, IHS-B--, or bijunctive, respectively.)
    \item If $\la \Gamma \ra \in \{\mathrm{IL_2},\mathrm{IL_3}\}$, then $\dcsp(\Gamma)$ is equivalent to \textsc{Nearest Codeword} and to \odds\ under A-reductions (note that these constraint languages are affine) and has no constant-factor FP-approximation, unless $\fpt=\W 1$.
    \item If $C \subseteq \la \Gamma \ra$, where $C \in
      \{\mathrm{IE_2}, \mathrm{IV}_2, \mathrm{IN}_2\}$, then
      $\dcsp(\Gamma)$ does not have an \fpa\ algorithm,
      unless $\fpt=\WP$. (Note that in these cases $\Gamma$ can
      $\exists \wedge$-define either arbitrary Horn relations, or
      arbitrary dual Horn relations, or the relation $\nae = \{0,1\}^3 \setminus \{(0,0,0),(1,1,1)\}$.)
  \end{enumerate}
\end{theorem}

 Looking at the co-clone lattice, it is easy to see that Theorem~\ref{trichotomy} covers all cases.
It is also easy to check that Theorem~\ref{th:mainintro} formulated in the introduction follows from Theorem~\ref{trichotomy}.
Theorem~\ref{trichotomy} is proved the following way.
Statement 1 is proved in Section~\ref{sec:csps-with-fpa} (Lemma~\ref{cor:01-valid}, and Corollaries~\ref{2sat} and \ref{IHSB}).
Statement 2 is proved in Section~\ref{affine_stuff} (Theorem~\ref{affine_equiv}) and in Section~\ref{sec:odd-set-probably}.
Statement 3 is proved in Section~\ref{sec:horn-dual-horn} (Corollary \ref{cor:IE2IV2} and Lemma \ref{IN2}).

\section{CSPs with \fpa\ algorithms}\label{sec:csps-with-fpa}

We prove the first statement of Theorem~\ref{trichotomy} by going through co-clones one by one.
As every relation of a $0$-valid \dcsp\ is always satisfied by the all $0$ assignment, and every relation of a $1$-valid \dcsp\ is always satisfied by the all $1$ assignment, we have a trivial algorithm for these problems.

\begin{lemma}\label{cor:01-valid}
  If $\la \Gamma \ra \subseteq \mathrm{II_0}$ or
  $\la \Gamma \ra \subseteq \mathrm{II_1}$, then $\dcsp(\Gamma)$ is polynomial-time solvable.
\end{lemma}

Consider now the co-clone \textup{ID}$_2$. \almost\ is defined as $\dcsp(\Gamma(\twosat))$, where  $\Gamma(\twosat) = \{x \vee y, x \vee \bar y, \bar x \vee \bar y\}$.
\begin{theorem}[\cite{almost2satfpt}]
  \textsc{Almost 2-SAT} is fixed-parameter tractable.
\end{theorem}
Since every bijunctive relation can be pp-defined by \twosat, the
constant-factor FP-approxi\-mability of bijunctive languages easily follows
from the FPT algorithm for \almost\ and from Corollary~\ref{FPA_spreads}.
\begin{corollary}\label{2sat}
If $\la \Gamma \ra \subseteq \mathrm{ID_2}$, then $\dcsp(\Gamma)$ has a constant-factor \fpa\ algorithm.
\end{corollary}
\begin{proof}
	We check in Table~\ref{cc_table} that $B = \{x \oplus y, x \rightarrow y\}$ is a base for the co-clone $\mathrm{ID_2}$. Relations in $B$ (and equality) can be $\exists \wedge$-defined over $\Gamma(\twosat)$, so the result follows from Corollary~\ref{FPA_spreads}.
\end{proof}

We consider now $\mathrm{IS_{00}}$ and $\mathrm{IS_{10}}$.
We first note that if $\la \Gamma \ra$ is in $\mathrm{IS_{00}}$ or
$\mathrm{IS_{10}}$, then the language is $k$-IHS-B+ or $k$-IHS-- for some $k \ge 2$.
\begin{lemma}\label{IHSBB0}
If $\la \Gamma \ra \subseteq \mathrm{IS_{00}}$, then there is an
integer $k \ge 2$ such that $\Gamma$ is $k$-IHS-B+.
If $\la \Gamma \ra \subseteq \mathrm{IS_{10}}$, then there is an
integer $k \ge 2$ such that $\Gamma$ is $k$-IHS-B--.
\end{lemma}
\begin{proof}
	We note that there is no finite base for $\mathrm{IS_{00}}$, thus
	$\la \Gamma \ra$ is a proper subset of $\mathrm{IS_{00}}$ (see
	Table~\ref{cc_table}). There is no proper base for $\text{IS}_{01}$,
	$\text{IS}_{02}$, or $\text{IS}_0$ either.  It follows from the
	structure of the co-clone lattice (Figure~\ref{classfigure}) that
	there exists a finite $k\ge 2$ such that
	$\la \Gamma \ra \subseteq \mathrm{IS_{00}}^k$, which implies that
	$\Gamma$ is IHS-B+. The proof of the second statement is analogous.
\end{proof}

By Lemma~\ref{IHSBB0}, if $\la \Gamma \ra \subseteq \mathrm{IS_{00}}$, then $\Gamma$ is generated by the relations $\bar x, x \rightarrow y, x_1 \vee \dots \vee x_k$ for some $k\ge 2$. The \dcsp\ problem for this set of relations is known to admit a constant-factor approximation.
\begin{theorem}[\cite{BooleanBook}, Lemma~7.29]\label{LP_approx}
  $\dcsp(\bar x, x \rightarrow y, x_1 \vee \dots \vee x_k)$ has a $(k+1)$-factor approximation algorithm (and hence has a constant-factor FPA algorithm).
\end{theorem}
Now Theorem~\ref{LP_approx} and Corollary~\ref{FPA_spreads} imply that
there is a constant-factor FPA algorithm for $\dcsp(\Gamma)$
whenever $\la \Gamma \ra$ is in the co-clone $\mathrm{IS_{00}}$ or
$\mathrm{IS_{10}}$ (note that equality can be \ea-defined using
$x\rightarrow y$). In fact, the resulting
algorithm is a polynomial-time approximation algorithm:
Theorem~\ref{LP_approx} gives a polynomial-time algorithm and this is
preserved by Corollary~\ref{FPA_spreads}.

\begin{corollary}\label{IHSB}
If $\la \Gamma \ra \subseteq \mathrm{IS_{00}}$ or $\la \Gamma \ra \subseteq \mathrm{IS_{10}}$, then $\dcsp(\Gamma)$ has a constant-factor \fpa\ algorithm.
\end{corollary}

Note that Theorem~7.25 in \cite{BooleanBook} gives a complete classification of Boolean \dcsp{}s with respect to constant-factor approximability. As mentioned, these \dcsp{}s also admit a constant-factor approximation algorithm. The reason we need Corollary~\ref{IHSB} is to have the characterization in terms of the co-clone lattice.

\section{CSPs equivalent to \odds}\label{affine_stuff}
In this section we show the equivalence of several problems under A-reductions. We identify CSPs that are equivalent to the following well-known combinatorial problems. In the \textsc{Nearest Codeword} (\textsc{NC}) problem, the input is an $m \times n$ $0/1$-matrix $A$, and an $m$-dimensional vector $b$. The output is an $n$-dimensional vector $x$ that minimizes the Hamming distance between $Ax$ and $b$.
In the \textsc{Odd Set} problem, the input is a set-system $\cS = \{S_1,S_2,\ldots,S_m\}$ over universe $U$. The output is a subset $T \subseteq U$ of minimum size such that every set of $\cS$ is hit an odd number of times by $T$, that is, $\forall i \in [m]$, $|S_i \cap T|$ is odd.

\eoset is the same problem as \odds, except that for each set we can specify whether it should be hit an even or odd number of times (the objective is the same as in \odds: find a subset of minimum size satisfying the requirements).
 We show that there is a cost-preserving reduction from \eoset to \odds.

\begin{lemma}\label{even_kaputt}
	There is a cost-preserving reduction from \eoset to \odds.
\end{lemma}
\begin{proof}
	Let $I$ be the instance of \eoset. If all sets in $I$ are even sets, then the empty set is an optimal solution. Otherwise, fix an arbitrary odd set $S_o$ in $I$. We obtain an instance $I'$ of \odds by introducing every odd set of $I$ into $I'$, and for each even set $S_e$ of $I$, we introduce the set $S_e \triangle S_o$ into $I'$, where $\triangle$ denotes the symmetric difference of two sets. This completes the reduction.
	
	Let $W$ be a solution of size $k$ for $I$. We claim that $W$ is also a solution of $I'$. Then those sets of $I'$ that correspond to odd sets of $I$ are obviously hit an odd number of times by $W$. We have to show that the remaining sets of $I'$ are also hit an odd number of times. Let $T = S_e \triangle S_o$ be such a set for some even set $S_e$ of $I$ and the fixed set $S_o$.	Let $A = S_o \setminus S_e$, $B = S_e \cap S_o$, and $C = S_o \setminus S_e$. If $W$ hits $A$ an even number of times, then $B$ must be hit an odd number of times (since $A \cup B = S_o$), and therefore $C$ is hit an odd number of times (as $B \cup C = S_e$). Since $T = S_e \triangle S_o = A \cup C$, $T$ is hit an odd number of times. If $W$ hits $A$ an odd number of times, then $B$ must be hit and even number of times, therefore $C$ must be hit an even number of times, and therefore $T$ is hit an odd number of times.
	
	Conversely, let $W'$ be a solution of size $k$ for $I'$. We claim that $W'$ is also a solution for $I$.  Odd sets of $I$ are clearly hit an odd number of times by $W'$.
	Let $S_e$ be an even set of $I$. We show that $S_e$ is hit an even number of times. The fixed set $S_o$ (also in the instance $I'$) and $S_e \triangle S_o$ are both hit an odd number of times by $W'$. Define sets $A,B,C$ as above. If $W'$ hits $C$ an odd number of times, then $A$ must be hit an even number of times, as $A \cup C = S_e \triangle S_o$ is hit an odd number of times. Since $S_o = B \cup A$ is hit an odd number of times, $B$ is hit an odd number of times. Since $S_e = C\cup B$, $S_e$ is hit an even number of times. The case when $W'$ hits $C$ an even number of times can be analyzed is similarly.	
\end{proof}

We define the relations $\even^m = \left \{(a_1,\dots,a_m)
  \in \{0,1\}^m : \sum_{i=1}^m a_i \text{ is even} \right \}$, $\odd^m
= \left \{(a_1,\dots,a_m) \in \{0,1\}^m : \sum_{i=1}^m a_i \text{ is
    odd} \right \}$, and the languages $B_2 = \{\even^4,x,\bar{x}\}$, $B_3 = \{\even^4,x
\oplus y\}$. Note that $B_2$ and $B_3$ are bases for the co-clones
$\mathrm{IL_2}$ and $\mathrm{IL_3}$, respectively.

\begin{theorem}\label{affine_equiv}
  The following problems are equivalent under cost-preserving reductions:
(1) \nearc, (2) \odds, (3) $\dcsp(B_2)$, and (4) $\dcsp (B_3)$.
\end{theorem}
\begin{proof}
	\noindent $\mathbf{(1) \Rightarrow (2)}$: Let $A$ be an $m \times n$ generator matrix for the \textsc{Nearest} \textsc{Codeword} problem, and $b$ be an $m$-dimensional vector such that we want to find a codeword of Hamming distance at most $k$ from $b$. Let $C$ be the set of all codewords generated by $A$, i.e., vectors in the column space of $A$. Let $A^{\perp}$ be the $\ell \times m$ matrix whose rows form a basis for the subspace perpendicular to the column space of $A$. Then $w \in C$ if and only if $A^{\perp} w = 0$. Assume now that $z$ is a vector that differs from $b$ at most in $k$ positions. Then we can write $z = z' + c$, where the weight of $z'$ is the distance between $z$ and $z'$. To find such a $z$, we write $A^{\perp} (z'+b) = 0$, and now we wish to find a solution that minimizes the weight of $z'$. Observe that $A^{\perp} z' = A^{\perp} c$ (since we are working in $\gf$). This can be encoded as a problem where we have a ground set $U = \{1,\dots,m\}$, and sets $S_i$, $1 \leq i \leq \ell$, defined as follows. Element $j$ is in $S_i$ if $A^{\perp}(i,j) = 1$. We want to find a subset $W \subseteq U$ of size at most $k$ such that $S_i$ is hit an even number of times if the $i$-th element of the vector $A^{\perp} b$ is $0$, and an odd number of times if it is $1$. This is an instance of the \eoset\ problem. Using Lemma~\ref{even_kaputt}, we can further reduce this problem to \odds, and we are done.
	
	$\mathbf{(2) \Rightarrow (3)}$: Note that Lemma~1 in \cite{Crowston13} can be adapted to obtain the reduction from \odds to $\dcsp(B_2)$. We show that there is a cost-preserving reduction from \textsc{Odd Set} to $\dcspx(B_2)$ (and hence to $\dcsp(B_2)$ by Lemma~\ref{lem:reducerestricted}).
	First we $\exists\wedge$-express the relation $\odd^n$ using $B_2$. We use an induction on $n$. For $n = 2$, we have that $\odd^2(x_1,x_2) = \even^4(x_1,x_2,0,1)$. Assume we have a formula that defines $\odd^n$. Then observe that
	\[\odd^{n+1}(x_1,\dots,x_{n+1}) = \exists u \odd^n(x_1,\dots,x_{n-1},u) \wedge \even^3(u,x_n,x_{n+1}).\]
	
	The variables of the \dcsp\ instance $J$ are the elements of the ground set of the \textsc{Odd Set} instance $I$.
	For each set $\{y_1,\dots,y_s\}$ of $I$, we add an undeletable constraint $\odd^s(y_1,\dots,y_s)$ to $J$. Finally, for each variable $y$ that appears in a constraint, we add the constraint $\bar{y}$. It is easy to see that a hitting set of size $k$ for $I$ corresponds to a deletion set of size $k$ for $J$ (consisting of constraints of the form $\bar{y}$).
	
	$\mathbf{(3) \Rightarrow (4)}$:	This follows from Lemma~\ref{adding_constants}.
	
	$\mathbf{(4) \Rightarrow (1)}$: Let $I$ be the $\dcsp(B_3)$ instance, and assume it has $n$ variables and $m$ constraints. We define a \textsc{Nearest Codeword} instance $J$ as follows. The matrix $A$ has dimension $m \times n$, and columns are indexed by the variables of $I$. If the $i$-th constraint of $I$ is $\even^4(x_{j_1}, x_{j_2}, x_{j_3},x_{j_4})$, then the $i$-th row of $A$ has $1$-s in positions $j_1,j_2,j_3,j_4$, and the $i$-th entry of vector $b$ is $0$. If the $i$-th constraint of $I$ is $x_{k_1} \oplus x_{k_2}$, then the $i$-th row of $A$ has $1$-s in positions $k_1$ and $k_2$, and the $i$-th entry of $b$ is $1$. Clearly, a deletion set of size $k$ for $I$ corresponds to a solution of $J$ having distance $k$ from vector $b$.
\end{proof}

\odds has the so-called \emph{self-improvement} property.
Informally, a polynomial time (resp. fixed-parameter time) approximation within some ratio $r$ can be turned into a polynomial time (resp. fixed-parameter time) approximation within some ratio close to $\sqrt{r}$.

\begin{lemma}\label{lem:self-improvement}
If there is an $r$-approximation for \odds running in time $f(n,m,k)$ where $n$ is the size of the universe, $m$ the number of sets, and $k$ the size of an optimal solution, then for any $\varepsilon > 0$, there is a $(1+\varepsilon)\sqrt{r}$-approximation running in time $\max(f(1+n+n^2,1+m+nm,1+k+k^2),O(n^{1+\frac{1}{\varepsilon}}m))$.
\end{lemma}

\begin{proof}
The following reduction is inspired by the one showing the self-improvement property of \textsc{Nearest Codeword} \cite{Arora97}.
Let $\cS = \{S_1,\ldots,S_m\}$ be any instance over universe $U=\{x_1,\ldots,x_n\}$.
Let $\varepsilon > 0$ be any real positive value and $k$ be the size of an optimal solution.
We can assume that $k \geqslant \frac{1}{\varepsilon}$ since otherwise one can find an optimal solution by exhaustive search in time $O(n^{1+\frac{1}{\varepsilon}}m)$.
We build the set-system $\cS' = \cS \cup \bigcup_{i \in [n], j \in [m]} S_j^i \cup \{\{e\}\}$ over universe $U'=U \cup \bigcup_{i,h \in [n]} \{x^i_h\} \cup \{e\}$, where $e$ is a new element, such that $S_j^i=\{e,x_i\} \cup \{x^i_h $ $|$ $x_h \in S_j\}$.
Note that the size of the new instance is squared.
We show that there is a solution of size at most $k$ to instance $\cS$ if and only if there is a solution of size at most $1+k+k^2$ to instance $\cS'$.

If $T$ is a solution to $\cS$, then $T' = \{e\} \cup T \cup \{x^i_h $ $|$ $x_i,x_h \in T\}$ is a solution to $\cS'$.
Indeed, sets in $\cS \cup \{\{e\}\}$ are obviously hit an odd number of times.
And, for any $i \in [n]$ and $j \in [m]$, set $S_j^i$ is hit exactly once (by $e$) if $x_i \notin T$, and is hit by $e$, $x_i$, plus as many elements as $S_j$ is hit by $T$; so again an odd number of times.
Finally, $|T'|=1+|T|+|T|^2$.

Conversely, any solution to $\cS'$ should contain element $e$ (to hit $\{e\}$), and should intersect $U$ in a subset $T$ hitting an odd number of times each set $S_i$ ($\forall i \in [m]$).
Then, for each $x_i \in T$, each set $S^i_j$ with $j \in [m]$ is hit exactly twice by $e$ and $x_i$.
Thus, one has to select a subset of $\{x^i_1,\ldots,x^i_n\}$ to hit each set of the family $\{S^i_1,\ldots,S^i_m\}$ an odd number of times.
Again, this needs as many elements as a solution to $\cS$ needs.
So, if there is a solution to $\cS'$ of size at most $1+k+k^2$, then there is a solution to $\cS$ of size at most $k$.
In fact, we will only use the weaker property that if there is a solution to $\cS'$ of size at most $k$, then there is a solution to $\cS$ of size at most $\sqrt{k}$.

Now, assuming there is an $r$-approximation for \odds running in time $f(n,m,k)$, we run that algorithm on the instance $\cS'$ produced from $\cS$.
This takes time $f(1+n+n^2,1+m+nm,1+k+k^2)$ and produces a solution of size $r(1+k+k^2)$.
From that solution, we can extract a solution $T$ to $\cS$ by taking its intersection with $U$.
Finally, $T$ has size smaller than $\sqrt{r(1+k+k^2)} \leqslant \sqrt{r}(k+1)=(1+\frac{1}{k})\sqrt{r}k \leqslant (1+\varepsilon)\sqrt{r}k$.
\end{proof}

Repeated application of the self-improvement property in Lemma~\ref{lem:self-improvement} shows that any constant-ratio approximation implies the existence of $(1+\varepsilon)$-approximation for arbitrary small $\varepsilon>0$.
In a similar way, we can show that polylogarithmic approximation implies the existence of logarithmic approximation.
\begin{corollary}\label{cor:self-improvement}
\begin{enumerate}
\item  If \odds admits an FPA algorithm with some ratio $r\ge 1$, then, for any $\varepsilon > 0$, it also admits an FPA algorithm with ratio $1+\varepsilon$ and
\item  If \odds admits an FPA algorithm with ratio $\log ^c k$, then it also admits an FPA algorithm with ratio $\log k$.
\end{enumerate}
\end{corollary}

\begin{proof}
We observe that for any $r' > 1$, there exists an $\varepsilon > 0$ such that $(1+\varepsilon)\sqrt{(1+\varepsilon)\sqrt{r'}} \leqslant \sqrt{r'}$.
Thus, applying twice the reduction of Lemma~\ref{lem:self-improvement}, we can improve any fixed-parameter $r'$-approximation to a fixed-parameter $\sqrt{r'}$-approximation.
Therefore, starting with an $r$-approximation, we can repeatedly apply the self-improvement property a constant number of times to obtain an FPA algorithm with ratio arbitrarily close to $1$. Similarly, starting with a $\log^c k$-approximation, applying the self-improvement property a constant number of times gives a $\log k$-approximation. Note that we are repeating the reduction of Lemma~\ref{lem:self-improvement} a constant number of times, hence the reduction is still a polynomial-time reduction.
\end{proof}

\section{Hardness of Odd Set}
\label{sec:odd-set-probably}
In this section, we show that \textsc{Odd Set} has no constant-factor FP-approximation unless $\W 1=\fpt$.
This implies, due to a recent result by Arnab et al. \cite{Arnab18,Arnab18b}, that \textsc{Even Set} is $\W 1$-hard under randomized reductions.
We even rule out for \textsc{Odd Set} an FP-approximation with any polylogarithmic ratio, under the same assumption.

The proofs in this section use linear algebra. We will need the following notation.
If $n,m,d,k$ are positive integers and $q$ is a prime power, then $\mathbb{F}_q^d$ denotes the $d$-dimension vector space over $\mathbb{F}_q$.
Each vector $\textbf{v}\in\mathbb{F}_q^d$ can be written as $\textbf{v}=(v_1,v_2,\ldots, v_d)$ with $v_i\in\mathbb{F}_q$ for all $i\in [d]$. We will denote by $\textbf{v}[i]$ the $i$-th coordinate of $\textbf{v}$.
Let $\textbf{1}_d:=(1,1,\ldots, 1)\in \mathbb{F}_q^d$ and $\textbf{0}_d:=(0,0,\ldots, 0)\in \mathbb{F}_q^d$.
Given $\textbf{a}=(a_1,a_2,\ldots, a_n)\in \mathbb{F}_q^n$ and $\textbf{b}=(b_1,b_2,\ldots, b_m)\in \mathbb{F}_q^{m}$, we write $\textbf{a}\circ\textbf{b}:=(a_1,a_2,\ldots, a_n,b_1,b_2,\ldots, b_m)\in\mathbb{F}_q^{n+m}$ for the concatenation of these two vectors.

\subsection{One side gap for the biclique problem}
Our inapproximability result for \textsc{Odd Set} builds on the recent W[1]-hardness and inapproximability results for \textsc{Biclique} by Lin~\cite{DBLP:conf/soda/Lin15}. The decision version of \textsc{Biclique} problem asks for a complete bipartite subgraph with $k$ vertices on each side. We consider the approximation verison of the problem where one side is fixed and the other side has to be maximized. Formally, we define the following gap version of the problem.
\npprob{$\textsc{Gap-Biclique}(s,\ell,h)$}{A bipartite graph $G= (L\dotcup R, E)$ with $n$ vertices and $s,\ell,h\in\mathbb{N}$ with $\ell<h$}{$s$}{Distinguish between the following cases:
\begin{itemize}
\item[(yes)] There exist $s$ vertices in $L$  with $h$ common neighbors.
\item[(no)] Any $s$ vertices in $L$  have at most $\ell$ common neighbors.
\end{itemize}
\vspace{-1em}

}
The following theorem is  the main result of Lin~\cite{DBLP:conf/soda/Lin15}.

\begin{theorem}[{\cite[Theorem~1.3]{DBLP:conf/soda/Lin15}}]\label{thm:gapbiclique}
There is a polynomial time algorithm $\mathbb A$ such that for every graph
$G$ with $n$ vertices and $k\in \mathbb N$ with $\lceil
n^{\frac{6}{k+6}}\rceil> (k+6)!$ and $6\mid k+1$ the algorithm $\mathbb A$
constructs a bipartite graph $H=(A \dotcup B, E)$ satisfying:
\begin{enumerate}
\item if $G$ contains a clique of size $k$, i.e., $K_k\subseteq G$, then
    there are $s$ vertices in $A$ with at least $\lceil
    {n^{\frac{6}{k+1}}}\rceil$ common neighbors in $B$;

\item otherwise $K_k\not\subseteq G$, any $s$ vertices in $A$ have at most
    $(k+1)!$ common neighbors in $B$,
\end{enumerate}
where $s=\binom{k}{2}$.
\end{theorem}

Our goal is to reduce \textsc{Gap-Biclique} to the following gap version of \textsc{Odd Set}.
\npprob{$\textsc{Gap-Odd-Set}(k_1,k_2)$}{A set $\mathcal{W}=\{w_1,w_2,\ldots,w_n\}$  of $n$ vectors from the vector space $\mathbb{F}_2^m$ ($m=n^{O(1)}$) and $k_1,k_2\in\mathbb{N}$ with $k_1\le k_2$}{$k_1$}{Distinguish between the following two cases: \begin{itemize}
\item[(yes)] There exists a set $I\subseteq [n]$ with $|I| \le k_1$ s.t.
$\sum_{i\in I}w_i=\textbf{1}_m$.
\item[(no)] For all sets $I\subseteq [n]$ with $|I|\le k_2$, $\sum_{i\in I}w_i\neq \textbf{1}_m$.
\end{itemize}
\vspace{-1em}

}
We prove that unless $\W 1=\fpt$, $\textsc{Gap-Odd-Set}(k,ck)$  has no \fpt-algorithm
 for any constant $c\ge 1$.
Our reduction takes two steps. First we show that the following problem $\textsc{Gap-Linear-Dependent-Set}_{q}(k_1,k_2)$ is  \W 1-hard for  $k_2= k_1\log k_1$ and $q=n^{O(1)}$. As finding dependent sets is a fairly natural algorithmic problem in linear algebra, the FP-inapproximability of this problem is interesting on its own right.
\npprob{$\textsc{Gap-Linear-Dependent-Set}_{q}(k_1,k_2)$}{A set $\mathcal{W}=\{w_1,w_2,\ldots,w_n\}$  of $n$ vectors from the vector space $\mathbb{F}_q^m$ ($m=n^{O(1)}$) and $k_1,k_2\in\mathbb{N}$ with $k_1\le k_2$}{$k_1$}{Distinguish between the following cases:
\begin{itemize}
\item[(yes)] There exist $k_1$ vectors in
${\mathcal W}$  that are linearly dependent.
\item[(no)] There are no $k_2$ vectors in $\mathcal W$ that are linearly dependent.
\end{itemize}
\vspace{-1em}

}
Then we provide a gap-preserving reduction from $\textsc{Gap-Linear-Dependent-Set}_{q}(k_1,k_2)$ to $\textsc{Gap-Odd-Set}(k,ck)$.

\subsection{From Biclique to Linear Dependent Set}\label{sec:bicliquetomdc}
The following theorem states our first reduction, which transfers the inapproximability of \textsc{Biclique} to \textsc{Linear Dependent Set}.
\begin{theorem}\label{theo:biclique2mdc}
Given an instance $G=(L\;\dot\cup\;R,E)$ of $\textsc{Gap-Biclique}(s,\ell,h)$, there is an algorithm that constructs a set $\mathcal W$ of vectors of $\mathbb{F}_{q}^m$ in time $(|L|+|R|)^{O(1)}$ with $q$ and $m$ being bounded by $(|L|+|R|)^{O(1)}$, such that
\begin{itemize}
\item \textit{(yes)} if $G$  is a yes-instance of $\textsc{Gap-Biclique}(s,\ell,h)$, then $\mathcal W$ contains a set of linearly dependent vectors of size equal to $sh$,
\item \textit{(no)} if $G$  is a no-instance of $\textsc{Gap-Biclique}(s,\ell,h)$, then every set of vectors from $\mathcal W$ of size at most $h\sqrt[s]{\frac{h}{\ell}}-1$ is linearly independent.
\end{itemize}

\end{theorem}
\medskip

Before proving Theorem~\ref{theo:biclique2mdc}, let us show how it can be used to prove an inapproximability result for \textsc{Gap-Linear-Dependent-Set}.
\begin{theorem}
There is no FPT-algorithm for $\textsc{Gap-Linear-Dependent-Set}_{q}(k_1,k_2)$ with $k_2=k_1\log k_1$ and $q=n^{O(1)}$, unless $\W 1=\fpt$.
\end{theorem}
\begin{proof}
We show how \textsc{$k$-Clique} can be solved using an algorithm for $\textsc{Gap-Linear-Dependent}$ $\textsc{-Set}_{n^{(1)}}(k_1,k_2)$.
Let $(G,k)$ be an input instance of \textsc{$k$-Clique} and $n$ be the number of vertices in $G$. Without loss of generality, we can assume that $\lceil n^{\frac{6}{k+6}}\rceil>(k+6)!$ and $6\mid k+1$.
Let $s=\binom{k}{2}$, $\ell=(k+1)!$ and $h=k^{2k^3}$. For sufficiently large $n$, we have $(k+1)!\le\ell<h\le n^{6/(k+1)}$.
 Applying Theorem~\ref{thm:gapbiclique} to $G$ and $k$, we obtain a bipartite graph $H$ in $|G|^{O(1)}$-time  such that 
  $G$ contains a $k$-clique  if and only if $(H,s,\ell,h)$ is a
 yes-instance  of  $\textsc{Gap-Biclique}(s,\ell,h)$.

Applying Theorem~\ref{theo:biclique2mdc} to $(H,s,\ell,h)$, we obtain a set $\mathcal W$ of vectors of $\mathbb{F}_q^m$ for some $q$ and $m$ bounded by $|G|^{O(1)}$. Let $k_1=hs$ and $k_2=k_1\log k_1$.
Note that for $k\ge 4$
\[
\sqrt[s]{\frac{h}{\ell}}
=\sqrt[s]{\frac{k^{2k^3}}{(k+1)!}}
\ge \sqrt[k^2]{\frac{k^{2k^3}}{k^{2k}}}
\ge {k^{2k-2/k}}
\ge k^7\ge 2k^5\log k+2k^2\log k
\ge s\log(hs).
\]
Thus $k_1\log k_1=hs\log(hs)\le h\sqrt[s]{\frac{h}{\ell}}$.
 By the completeness and soundness of Theorem~\ref{theo:biclique2mdc}, we have $(\mathcal W,k_1,k_2)$ is a yes-instance of $\textsc{Gap-Linear-Dependent-Set}_q(k_1,k_2)$ if and only if
$G$ contains a $k$-clique.
\end{proof}
Let us prove now Theorem~\ref{theo:biclique2mdc}, the reduction from \textsc{Gap-Biclique} to \textsc{Gap-Linear-Dependent-Set}.

\begin{proof}[Proof (of Theorem~\ref{theo:biclique2mdc})]
 Assume that an instance $G=(L\;\dot\cup\;R,E)$ of $\textsc{Gap-Biclique}(s,\ell,h)$ is given. We construct a set $\mathcal W$ of vectors and prove the required completeness and soundness.

\textbf{Construction of $\mathcal W$.}
Let $q:=2^{\lceil\log (|L|+|R|)\rceil}\ge|L|+|R|$.  We identify $L$ and $R$ with disjoint subsets of $\mathbb{F}_q$. Let $B:=\max\{h,s\}-1$. We first define a function $\iota : L\cup R\to \mathbb{F}_q^B$ as follows.

\begin{itemize}
\item for each $v\in R$, $\iota(v):=(1,v,\ldots ,v^{h-2}) \circ \textbf{0}_{B-h+1}$,
\item for each $u\in L$, $\iota(u):=(1,u,\ldots ,u^{s-2}) \circ \textbf{0}_{B-s+1}$.
\end{itemize}

Using well-known properties of Vandermonde matrices, we can see that any $h-1$ vectors in $\iota(R)$ are linearly independent and any $h$ vectors from $\iota(R)$ are linearly dependent.
To summarize, $R$ and $\iota$ satisfy the following conditions:

\begin{itemize}
\item[(R1)] for all $I\in\binom{R}{h}$,  the vectors $\{\iota(v) : v\in I\}$ are linearly dependent.
\item[(R2)] for all  $I\in\binom{R}{h-1}$, the vectors $\{\iota(v) : v\in I\}$ are linearly independent.
\end{itemize}

Similarly, we also have that

\begin{itemize}
\item[(L1)] for all $I\in\binom{L}{s}$,  the vectors $\{\iota(u) : u\in I\}$ are linearly dependent.
\item[(L2)] for all  $I\in\binom{L}{s-1}$, the vectors $\{\iota(u) : u\in I\}$ are linearly independent.
\end{itemize}

Then we let $m=qB$ and  consider  vectors from $\mathbb{F}_q^m=\mathbb{F}_q^{qB}$, which can be seen as the concatenation of $q$ blocks, each of $B$ coordinates.
For $x\in \mathbb{F}_q^m$, we use the notation $x^{(i)}$ to refer  to the $i$-block, which is the $B$-dimensional vector given by coordinates $\{(i-1)B+1,(i-1)B+2,\ldots,(i-1)B+B)\}$.  For each $u\in L$ and $v\in R$ with $\{u,v\}=e\in E$, we introduce a vector $w_{e}\in \mathbb{F}_q^{qB}$ such that

\begin{itemize}
\item[(W1)] for all $i\in [q]\setminus\{v,u\}$, $w_e^{(i)}=\textbf{0}_B$,
\item[(W2)] $w_e^{(v)}=\iota(u)$,
\item[(W3)] $w_e^{(u)}=\iota(v)$.
\end{itemize}

That is, we can imagine $w_e$ as being partitioned $q$ blocks of $B$ coordinates, with the representation of $u$ appearing in the $v$-th block and the representation of $v$ appearing in the $u$-th block. Note the use of $u$ and $v$ in the definition: the $v$-th block on its own describes both $v$ (by its position) and $u$ (by its content), and similarly the $u$-th block also describes both endpoints of $e$. Finally, let
\[
\mathcal W:=\{w_e : e\in E\}.
\]
Obviously, $\mathcal W$ can be computed in time $(|L|+|R|)^{O(1)}$.

\textbf{(yes) case.} Suppose $G$ is a yes-instance of $\textsc{Gap-Biclique}(s,\ell,h)$. There exist a set $X\in\binom{L}{s}$ and a set $Y\in\binom{R}{h}$ such that for all
$u\in X$ and $v\in Y$, $\{u,v\}\in E$.
Suppose $X=\{u_1,u_2,\ldots,u_s\}$ and $Y=\{v_1,v_2,\ldots,v_h\}$. By (R1) and (L1), there exists $a_i\in \mathbb{F}_q$ for each $i\in [s]$ and $b_j\in \mathbb{F}_q$ for each $j\in [h]$ such that
\[
\sum_{i\in [s]}a_i\iota(u_i)=\textbf{0}_B
\text{ and }
\sum_{j\in [h]}b_j\iota(v_j)=\textbf{0}_B.
\]
By (R2) and (L2), we deduce that for every $i\in[s]$ and $j\in [h]$, $a_i\neq 0$ and $b_j\neq 0$.
We prove that ${\mathcal W}$ contains a set of $sh$ dependent vectors by showing
\[
\sum_{i\in[s],j\in[h]}a_ib_jw_{\{u_i,v_j\}}=\textbf{0}_m.
\]
Let $w=\sum_{i\in[s],j\in[h]}a_i b_j w_{\{u_i, v_j\}}$. It is easy to check that

\begin{itemize}
\item by (W1), for every $z\in [q] \setminus (X\cup Y)$, $w^{(z)}=\textbf{0}_B$,
\item by (W2), for every $v_j\in Y$, $w^{({v_j})}=\sum_{i\in [s]}a_ib_j\iota(u_i)=b_j\sum_{i\in [s]}a_i\iota(u_i)=\textbf{0}_B$,
\item by (W3), for every $u_i\in X$, $w^{(u_i)}=\sum_{j\in [h]}a_ib_j\iota(v_j)=a_i\sum_{j\in [h]}b_j\iota(v_j)=\textbf{0}_B$.
\end{itemize}
Note that $a_ib_j\neq 0$ for all $i\in[s]$ and $j\in[h]$.
It follows that $\{w_{\{u,v\}}:u\in X\text{ and } v\in Y\}$ is a set of $sh$ linearly dependent vectors from $\mathcal W$.

\textbf{(no) case.} Suppose $G$ is a no-instance of
 $\textsc{Gap-Biclique}(s,\ell,h)$. Let $W\subseteq\mathcal W$ be a set of  vectors that are linearly dependent. We define two vertex sets and their edge set as follows. Let
 \[
 X:=\{u\in L : \text{there exists $v\in R$ such that $w_{\{u,v\}}\in W$}\},
 \]
 \[
 Y:=\{v\in R : \text{there exists $u\in L$ such that $w_{\{u,v\}}\in W$}\},
 \]
 and
 \[
E_W:=\{ \{ u,v\}:  \text{$w_{\{u,v\}}\in W$ }\}.
 \]
First, we note that $X$ and $Y$ are not empty because $W$ is non-empty.
By (R2) and (W3), for every $u\in X$, there exist at least $h$ vertices in $Y$ that are adjacent to $u$, i.e. $|N(u)\cap Y|\ge h$. Similarly, by (L2) and (W2), for every $v\in Y$, we have $|N(v)\cap X|\ge s$.

\begin{claim}\label{cl:bi}
For $s,\ell,h\in\mathbb{N}^+$ and two non-empty sets $X$ and $Y$, let $(X\cup Y,E_W)$ be a bipartite graph such that:
\begin{itemize}
\item (i) every vertex in $X$ has at least $h$ neighbors,
\item (ii) every vertex in $Y$ has at least $s$ neighbors,
\item (iii) every $s$-vertex set of $X$ has at most $\ell$ common neighbors.
\end{itemize}
Then, $|E_W|\ge(\frac{h}{\ell})^{1/s}h$.
\end{claim}
\begin{proof}[Proof of Claim~\ref{cl:bi}]
Because $X$ is not empty, there exists a vertex $u\in X$. By (i), $u$ has at least $h$ neighbors in $Y$, so $|Y|\ge h$.
By (ii), for every $v\in Y$, $v$ has at least $s$ neighbors in $X$. If $\binom{|X|}{s}\ell<|Y|$, then there must exist a $s$-vertex set in $X$ which has more than $\ell$ common neighbors in $Y$. Thus we must have that
\[
|X|^s\ge \binom{|X|}{s}\ge \frac{|Y|}{\ell}\ge \frac{h}{\ell}.
\]
By (i) again, we conclude that $|E_W|\ge h|X|\ge(\frac{h}{\ell})^{1/s}h$, what we had to show.
\renewcommand{\qedsymbol}{$\lrcorner$}\end{proof}
As $E_W$ satisfies the conditions of Claim~\ref{cl:bi}, we have $|W|=|E_W|\ge h\sqrt[s]{\frac{h}{\ell}}$.
\end{proof}

\begin{remark}\label{rem:yescase}
Our construction produces instances $\mathcal W$ of \textsc{Gap-Linear-Dependent}-$\textsc{Set}_{n^{O(1)}}(k,k\log k)$ such that in the (yes) case, there exist exactly $k$ vectors $\textbf{v}_1,\ldots,\textbf{v}_k$ in $\mathcal W$ and $k$ nonzero elements $c_1,\ldots,c_k$ in $\mathbb{F}_{q}$ such that $\sum_{i\in[k]}c_i\textbf{v}_i=\textbf{0}_m$.
\end{remark}
\subsection{From  Linear Dependent Set to Odd Set}
It will be convenient to work with a colored version of \textsc{Gap-Linear-Dependent-Set}.

\npprob{$\textsc{Gap-Linear-Dependent-Set}_{q}^{col}(k_1,k_2)$}{A set $\mathcal{W}=\{w_1,w_2,\ldots,w_n\}$  of $n$ vectors from the vector space $\mathbb{F}_q^m$ ($m=n^{O(1)}$),  $k_1,k_2\in\mathbb{N}$ with $k_1\le k_2$ and a coloring $c :\mathcal W\to [k_1]$}{$k_1$}{
Distinguish between the following cases:
\begin{itemize}
\item[(yes)] There exist at most $k_1$ vectors in
${\mathcal W}$  with distinct colors under $c$ that are linearly dependent.
\item[(no)] There are no $k_2$ vectors in $\mathcal W$ that are linearly dependent.
\end{itemize}
\vspace{-1em}
}

With a standard application of the color-coding technique (see, e.g., \cite{DBLP:books/sp/CyganFKLMPPS15}), we can reduce an instance of $\textsc{Gap-Linear-Dependent-Set}_{q}(k_1,k_2)$ to $2^{O(k_1)}\log n$ instances of $\textsc{Gap-Linear-Dependent-Set}_{q}^{col}(k_1,k_2)$, showing the hardness of the latter problem.
\begin{theorem}\label{thm:gaplds_col}
There is no FPT-algorithm for $\textsc{Gap-Linear-Dependent-Set}_{q}^{col}(k_1,k_2)$ with $k_2=k_1\log k_1$ and $q=n^{O(1)}$, unless $\W 1=\fpt$.
\end{theorem}

We present the following reduction as a warm up. Together with Theorem~\ref{thm:gaplds_col}, it shows that there is no better than 3-approximation for \textsc{Odd Set}, unless $\W 1=\fpt$.
\begin{theorem}\label{th:oddsetred1}
For  $d,k\in\mathbb{N}^+$ with $k\ge 8$, given an instance  $\mathcal W$ of \textsc{Gap-Linear-Dependent}-$\textsc{Set}_{2^d}^{col}(k,k\log k)$ with $\mathcal W\subseteq \mathbb{F}_{2^d}^m$ one can construct an instance $\mathcal W'$ of \textsc{Gap-Odd-Set}$(k+1,3k+1)$ with $\mathcal W'\subseteq \mathbb{F}_2^{1+md+k}$ in $(2^d|{\mathcal W}|)^{O(1)}$ time such that if $\mathcal W$ is a yes-instance (resp. no-instance) then $\mathcal W'$ is a yes-instance (resp. no-instance).
\end{theorem}

\begin{proof}
Suppose $\mathcal W$ is an instance of \textsc{Gap-Linear-Dependent}-$\textsc{Set}_{2^d}^{col}(k,k\log k)$ with the coloring $c$. For the definition of the set $\mathcal W'$, we need to introduce some notations.
Let $\eta_i$ be the $k$-dimensional vector with $1$ at the $i$-th position and $0$ everywhere else.
We can view the finite field $\mathbb{F}_{2^d}$ as a $d$-dimensional vector space of $\mathbb{F}_2$, hence there exist $d$ elements $e_1,e_2,\ldots,e_d$ in $\mathbb{F}_{2^d}$ such that every $v\in \mathbb{F}_{2^d}$ can be expressed, in a unique way, as the sum of a subset of the $e_i$'s.
Fix $d$ such elements $e_1,e_2,\ldots,e_d$ and let $f : \mathbb{F}_{2^d}\to \mathbb{F}_2^{d}$ be a function that for each $v=\sum_{i\in [d]} c_ie_i\in \mathbb{F}_{2^d}$, $f(v)=(c_1,c_2,\ldots,c_d)\in \mathbb{F}_2^d$. Similarly,  we can define a function $F : \mathbb{F}_{2^d}^m\to \mathbb{F}_2^{md}$ with $F(v)=f(v[1])\circ \dots \circ  f(v[m])$.
Observe that for any $x\in \mathbb{F}_{2^d}^m$, we have $F(x)=\textbf{0}_{md}$ if and only if $x=\textbf{0}_m$, and $F(x+y)=F(x)+F(y)$ holds for any $x,y\in  \mathbb{F}_{2^d}^m$.

The set $\mathcal W'$ is defined as follows.
\[
\mathcal W':=\{\textbf{0}_1\circ F(aw)\circ \eta_{c(w)} :w\in\mathcal W, a\in\mathbb{F}_{2^d}^+\}\cup\{\textbf{1}_1\circ\textbf{1}_{md}\circ\textbf{0}_k\}.
\]

\begin{itemize}
\item (yes)
Suppose $\mathcal W$ is a yes-instance of \textsc{Gap-Linear-Dependent}-$\textsc{Set}_{2^d}^{col}(k,k\log k)$. Without loss of generality,  we can assume that there exist $w_1,w_2,\ldots,w_k\in\mathcal W$ and $a_1,a_2,\ldots,a_k\in\mathbb{F}_{2^d}^+$ such that for all $i\in[k]$ $c(w_i)=i$ and
\[
\sum_{i\in[k]}a_iw_i=\textbf{0}_m.
\]
It is easy to verify that
\[
\textbf{1}_1\circ\textbf{1}_{md}\circ\textbf{0}_k+\sum_{i\in[k]}(\textbf{0}_1\circ F(a_iw_i)\circ \eta_{i})=\textbf{1}_{1+md+k}.
\]

\item (no)
Suppose $\mathcal W$ is a no-instance and $W'\subseteq\mathcal W'$ is a set of vectors whose sum $w'$ is equal to $\textbf{1}_{1+md+k}$. First, $W'$ must contain the vector $\textbf{1}_1\circ\textbf{1}_{md}\circ\textbf{0}_k$, otherwise $w'$ has $0$ in its first coordinate. For each $i\in [k]$, let
\[
Y_i:=\{\textbf{0}_1\circ F(aw)\circ\eta_{i}\in W' : \text{where $w\in\mathcal W$, $a \in\mathbb{F}_{2^d}^+$ and $c(w)=i$}\}.
\]
From the definition of $Y_i$, we deduce that for each $i\in [k]$, $|Y_i|$ must be odd, otherwise the $(1+md+i)$-th element of  $w'$ is not equal to $1$. 

Let
\[
X_i:=\{(w,a) : \text{$w\in\mathcal W$, $a \in\mathbb{F}_{2^d}^+$, $c(w)=i$ and $\textbf{0}_1\circ F(aw)\circ\eta_{i}\in W'$}\}.
\]
It is not hard to verify that 
\[
Y_i=\{\textbf{0}_1\circ F(aw)\circ\eta_{i}: \text{$(w,a)\in X_i$}\}.
\]
Since $\mathcal W$ is a no-instance and $k>2$, for any $a_1,a_2\in F_{2^d}^+$ and two distinct $w_1,w_2$ in $\mathcal W$, $F(a_1w_1)\neq F(a_2w_2)$. We deduce that 
\[
|Y_i|=|X_i|.
\]
Note that $W'=\bigcup_{i\in[k]}Y_i\cup\{\textbf{1}_1\circ\textbf{1}_{md}\circ\textbf{0}_k\}$.
If for all $i\in [k]$, $|X_i|\ge 3$, then $|W'|\ge 3k+1$ and we are done. Otherwise suppose for some $i\in [k]$, $X_i=\{(x^*,a^*)\}$. For each $w\in \mathcal W$, let $a_w:=\sum_{a\in\mathbb{F}_{2^d}^+,(w,a)\in X_{c(w)}}a$. Then we define
\[
X:=\{w : w\in\mathcal W, a_w\neq 0\}.
\]
Note that $x^*\in X$ so $X$ is not empty. Again, from the fact that two distinct $w_1,w_2$ in $\mathcal W$ are mapped to distinct $F(a_1w_1),F(a_2w_2)$ for $a_1,a_2\in\mathbb{F}_{2^d}^+$, we deduce that $|X|\le \sum_{i\in[k]}|Y_i|\le |W'|-1$.
Let $y$ be the sum of vectors in $W'\setminus\{\textbf{1}_1\circ\textbf{1}_{md}\circ\textbf{0}_k\}$.
Thus $y=w'-\textbf{1}_1\circ\textbf{1}_{md}\circ\textbf{0}_k=\textbf{1}_{1+md+k}-\textbf{1}_1\circ\textbf{1}_{md}\circ\textbf{0}_k=\textbf{0}_1\circ
\textbf{0}_{md}\circ\textbf{1}_k$. This means that $\sum_{x\in X}F(a_x\cdot x)=\textbf{0}_{md}$, which is only possible if $\sum_{x\in X}a_x\cdot x=\textbf{0}_m$.  Since $X$ is a non-empty set of linearly dependent vectors in $\mathcal W$, we must have $k\log k\le |X|\le |W'|-1$. Therefore $|W'|\ge 1+k\log k\ge 1+3k$.
\end{itemize}
\end{proof}

The self-improvement property of \textsc{Odd Set} (Corollary~\ref{cor:self-improvement}(1)) shows that a $c$-approximation for \textsc{Odd Set} implies the existence of a polynomial-time $c'$-approximation for every $c'>0$. This  means that the reduction in Theorems~\ref{th:oddsetred1} and \ref{thm:gaplds_col} already show that there is no constant-factor FP-approximation for \textsc{Odd Set}, unless $\W 1=\fpt$. However, with a slight change in the reduction, we can improve the inapproximability to $\log k$ and then Corollary~\ref{cor:self-improvement}(2) can improve this further to polylogarithmic inapproximability, proving Theorem~\ref{th:main-odd-set}.

\begin{theorem}
Given an instance $\mathcal W$ of $\textsc{Gap-Linear-Dependent-Set}_{2^d}^{col}(k,k\log k)$ one can construct $(d+1)^k$ instances of $\textsc{Gap-Odd-Set}(k,k\log k)$ in time $d^k n^{O(1)}$ such that $\mathcal W$ is a yes-instance if and only if at least one of the instances $\textsc{Gap-Odd-Set}(k+1,k\log k+1)$ is a yes-instance.
\end{theorem}
\begin{proof}
We use the same definitions for $\eta$, $f$, and $F$ as in the proof of Theorem~\ref{th:oddsetred1}.
For each $i\in [d]$, we define a subset $C_i$ of $\mathbb{F}_{2^d}^+$ by setting
\[
C_i:=\{c\in\mathbb{F}_{2^d}^+: \text{$f(c)[i]=1$ and $f(c)$ contains an even number of ones}\}.
\]
In addition, we put $C_{d+1}:=\{c\in\mathbb{F}_{2^d}^+: \text{$f(c)$ contains an odd number of ones}\}$. The sets $C_i$'s have the following properties that are crucial to our reduction.

\begin{itemize}
\item (i) $\bigcup_{i\in [d+1]}C_i=\mathbb{F}_{2^d}^+$.
\item (ii) for all $i\in [d+1]$ the sum of any odd numbers of elements in $C_i$ is not equal to $0$.
\end{itemize}

The main idea is the following. If there are $k$ vectors in $\mathcal{W}$ that can generate the zero vector with coefficients $a_1$, $\dots$, $a_k$, then we ``guess'' for each $a_i$ one of the sets $C_1$, $\dots$, $C_{d+1}$ that contains it. We create one instance for each of these guesses. The advantage of this approach is that we can more efficiently argue about soundness than in the proof of Theorem~\ref{th:oddsetred1}.

Formally, for every $g\in [d+1]^k$, we construct an instance $\mathcal W_g$ of $\textsc{Gap-Odd-Set}(k+1,k\log k+1)$ by setting
\[
\mathcal W_g:=\{\textbf{0}_1\circ F(aw)\circ\eta_{c(w)} :w\in\mathcal W,a\in C_{g[c(w)]}\}\cup\{\textbf{1}_1\circ\textbf{1}_{md}\circ\textbf{0}_k\}.
\]
\begin{itemize}
\item (yes)
Suppose $\mathcal W$ is a yes-instance of \textsc{Gap-Linear-Dependent}-$\textsc{Set}_{2^d}^{col}(k,k\log k)$. Without loss of generality, assume that there  exist $w_1,w_2,\ldots,w_k\in\mathcal W$ and $a_1,a_2,\ldots,a_k\in\mathbb{F}_{2^d}^+$ such that for all $i\in[k]$ $c(w_i)=i$ and
\[
\sum_{i\in[k]}a_iw_i=0.
\]
It follows that
\[
\textbf{1}_1\circ\textbf{1}_{md}\circ\textbf{0}_k+\sum_{i\in[k]}(\textbf{0}_1\circ F(a_iw_i)\circ \eta_{i})=\textbf{1}_{1+md+k}.
\]
Next we define $g^*\in [d+1]^k$ by setting  $g^*[i]:=\min\{j\in [d+1]: a_i\in C_j\}$ for all $i\in [k]$. Because $\bigcup_{j\in [d+1]}C_j=\mathbb{F}_{2^d}^+$, the vector $g^*$ is well defined. It is easy to check that for all $i\in [k]$,  $a_i\in C_{g^*[i]}$, hence
$\textbf{0}_1\circ F(a_iw_i)\circ \eta_{i}\in\mathcal W_{g^*}$. This means that $\mathcal W_{g^*}$ is a yes-instance.

\item (no)
Suppose $\mathcal W$ is a no instance and  $W'\subseteq\mathcal W_g$ for some $g\in [d+1]^k$ is a set of vectors whose sum $w'$ is equal to $\textbf{1}_{1+md+k}$. First, $W'$ must contain the vector $\textbf{1}_1\circ\textbf{1}_{md}\circ\textbf{0}_k$, otherwise $w'$ has $0$ in its first coordinate. For each $i\in [k]$, let
\[
X_i:=\{(w,a) : \text{$w\in\mathcal W$ and $a\in\mathbb{F}_{2^d}^+$ such that $c(w)=i$ and $\textbf{0}_1\circ F(aw)\circ\eta_{i}\in W'$}\}.
\]
It follows that for each $i\in [k]$, $|X_i|$ must be odd, otherwise the $(1+md+i)$-th element of  $w'$ is not equal to $1$.
For each $w\in \mathcal W$, let $a_w:=\sum_{a\in\mathbb{F}_{2^d}^+,(w,a)\in X_{c(w)}}a$. Then we define
\[
X:=\{w : w\in\mathcal W, a_w\neq 0\}.
\]
Because $|X_i|$ is odd for all $i\in [k]$, there exists an element $w_i\in \mathcal W$ such that $|\{a\in\mathbb{F}_{2^d}^+:(w_i,a)\in X_i\}|$ is odd. Note that for all $(w,a)\in X_i$, $a\in C_{g(i)}$.
By (ii), the sum of odd number of elements from  $C_{g(i)}$ is non-zero.
Thus $a_{w_i}=\sum_{a\in\mathbb{F}_{2^d}^+,(w_i,a)\in X_i}a\neq 0$, which implies that $X$ is not empty.

Let $y$ be the sum of vectors in $W'\setminus\{\textbf{1}_1\circ\textbf{1}_{md}\circ\textbf{0}_k\}$.
Thus $y=w'-\textbf{1}_1\circ\textbf{1}_{md}\circ\textbf{0}_k=\textbf{1}_{1+md+k}-\textbf{1}_1\circ\textbf{1}_{md}\circ\textbf{0}_k=\textbf{0}_1\circ
\textbf{0}_{md}\circ\textbf{1}_k$. This means that $\sum_{x\in X}F(a_x\cdot x)=\textbf{0}_{md}$, which is only possible if $\sum_{x\in X}a_x\cdot x=\textbf{0}_m$.  Since $X$ is a non-empty set of linearly dependent vectors in $\mathcal W$, we must have $k\log k\le |X|\le |W'|-1$. Therefore $|W'|\ge 1+k\log k$.
\end{itemize}
\end{proof}

\section{Hardness of Horn ($\textup{IE}_2$), dual-Horn ($\textup{IV}_2$) and $\text{IN}_2$}\label{sec:horn-dual-horn}
In this section, we establish statement 3 of Theorem~\ref{trichotomy} by proving the inapproximability of $\dcsp(\Gamma)$ if $\Gamma$ generates
one of the co-clones $\text{IE}_2$, $\text{IV}_2$, or
$\text{IN}_2$. The inapproximability proof uses previous results on
the inapproximability of circuit satisfiability problems.

A \emph{monotone Boolean circuit} is a directed acyclic graph, where each node with in-degree at least $2$ is labeled as either an AND node or as an OR node, each node of in-degree 0 is an input node, and nodes with in-degree 1 are not allowed. Furthermore, there is a node with out-degree $0$ that is the output node. Let $C$ be a monotone Boolean circuit. Given an assignment $\varphi$ from the nodes of $C$ to $\{0,1\}$, we say that $\varphi$ satisfies $C$ if
\begin{itemize}
\item for any OR node $G_\vee$ with in-neighbors $G_1,\dots,G_n$, $\varphi(G_{\vee}) = 1$ if and only if $1 \in \{\varphi(G_1),\dots,\varphi(G_n)\}$,
\item for any AND node $G_\wedge$ with in-neighbors $G_1,\dots,G_n$, $\varphi(G_{\vee}) = 1$ if and only if $\varphi(G_1) = \dots = \varphi(G_n) = 1$, and
\item $\varphi(G_o) = 1$, where $G_o$ is the output node.
\end{itemize}
The \emph{weight} of an assignment is the number of input nodes with value $1$. Circuit $C$ is $k$-satisfiable if there is a weight-$k$ assignment satisfying $C$. The problem \textsc{Monotone Circuit Satisfiability (MCS)} takes as input a monotone circuit $C$ and an integer $k$, and the task is to decide if there is a satisfying assignment of weight at most $k$. The following theorem is a restatement of a result of Marx \cite{antimonotone}. We use this to show that Horn-CSPs are hard.
\begin{theorem}[\cite{antimonotone}]\label{monotone_bound}
\textsc{Monotone Circuit Satisfiability} does not have an \fpa\ algorithm, unless $\fpt = \WP$.
\end{theorem}

The following corollary can be achieved by simply replacing nodes of in-degree larger than 2 with a circuit
consisting of nodes having in-degree at most two in the standard way.
\begin{corollary}\label{monotone_bound_2}
\textsc{Monotone Circuit Satisfiability}, where circuits are restricted to have nodes of in-degree at most $2$, does not have an FPA algorithm, unless $\fpt = \WP$.
\end{corollary}

We use Corollary~\ref{monotone_bound_2} to establish the inapproximability of Horn-SAT and dual-Horn-SAT, assuming that $\fpt \neq \WP$. Using the co-clone lattice, this will show hardness of approximability of $\dcsp(\Gamma)$ if $\la \Gamma \ra \in \{\text{IV}_2, \text{IE}_2\}$.

\begin{lemma}\label{monotone_reduction}
  If \dcsp$(\{x \vee y \vee \bar{z}, x, \bar{x}\})$ or \dcsp$(\{\bar{x} \vee \bar{y} \vee z, x, \bar{x}\})$ has a constant-factor FP-approximation, then $\fpt = \WP$.
\end{lemma}

\begin{proof}
	We prove that there is a cost-preserving polynomial-time reduction from \textsc{Monotone Circuit Satisfiability} to $\dcspx(\{x \vee y \vee \bar{z},x,\bar{x}\})$. This is sufficient by Corollary~\ref{monotone_bound_2}. Let $C$ be an \textsc{MCS} instance. We produce an instance $I$ of \dcspx\ as follows. For each node of $C$, we introduce a new variable into $I$, and we let $f$ denote the natural bijection from the nodes of $C$ to the variables of the instance $I$.
	
	We add undeletable constraints to simulate the computation of each AND node of $C$ as follows. Observe first that the implication relation $x \rightarrow y$ can be expressed as $y \vee y \vee\bar{x}$. For each AND node $G_\wedge$ such that $G_1$ and $G_2$ are the nodes feeding into $G_\wedge$, we add two constraints to $I$ as follows. Let $y = f(G_{\wedge})$, $x_1 = f(G_1)$, and $x_2 = f(G_2)$. We place the constraints $y \rightarrow x_1$ and $y \rightarrow x_2$ into $I$. We observe that the only way variable $y$ could take on value $1$ is if both $x_1$ and $x_2$ are assigned $1$. (In this case, note that $y$ could also be assigned $0$ but that will be easy to fix.)
	
	Similarly, we add constraints to simulate the computation of each OR node of $C$ as follows. For each OR node $G_\vee$ such that $G_1$ and $G_2$ are the nodes feeding into $G_\vee$, we add the constraint $x_1 \vee x_2 \vee \bar{y}$ to $I$, where $y = f(G_{\vee})$, $x_1 = f(G_1)$, and $x_2 = f(G_2)$. Note that if both $x_1$ and $x_2$ are $0$, then $y$ is forced to have value $0$. (Otherwise $y$ can take on either value $0$ or $1$, but again, this difference between an OR function and our gadget will be easy to handle.)
	
	In addition, we add a constraint $x_o = 1$, where $x_o$ is the variable such that $x_o = f(G)$, where $G$ is the output node. All constraints that appeared until now are defined as undeletable (recall that \dcspx\ allows undeletable constraints). To finish the construction, for each variable $x$ such that $x = f(G)$ where $G$ is an input node, we add a constraint $x = 0$ to $I$. We call these constraints \textit{input constraints}. Note that only input constraints can be deleted in $I$.
	
	If there is a satisfying assignment $\varphi_C$ of $C$ of weight $k$, then we delete the input constraints $x = 0$ of $I$ such that $\varphi_C(G) = 1$, where $G$ is a node such that $f(G) = x$. Clearly, the map $\varphi_C \circ f^{-1}$ is a satisfying assignment for $I$, where we needed $k$ deletions.
	
	For the other direction, assume that we have a satisfying assignment $\varphi_I$ for $I$ after removing some $k$ input constraints. We repeatedly apply the following two modifications of $\varphi_I$ as long as possible.
        This process terminates since $C$ is a directed acyclic graph.
        Let $x_1,x_2$, and $y$ be variables such that $f^{-1}(x_1)$ and $f^{-1}(x_2)$ are in-neighbors of node $f^{-1}(y)$.
\begin{enumerate}
\item If $f^{-1}(y)$ is an AND node, $\varphi_I(x_1) = 1, \varphi_I(x_2) = 1$, and $\varphi_I(y) = 0$, then we change $\varphi_I(y)$ to $1$.
\item If $f^{-1}(y)$ is an OR node, $1 \in \{\varphi_I(x_1), \varphi_I(x_2)\}$, and $\varphi_I(y) = 0$, then we change $\varphi_I(y)$ to $1$.
\end{enumerate}
It follows from the definition of the constraints we introduced for AND and OR nodes that once we finished modifying $\varphi_I$, the resulting assignment $\varphi_I'$ is still a satisfying assignment. Now it follows that $\varphi_I' \circ f$ is a weight $k$ satisfying assignment for $C$.

To show the inapproximability of \dcsp(\{$\bar{x} \vee \bar{y} \vee z, x, \bar{x}\}$), we note that there is a cost-preserving bijection between instances of \dcsp(\{$\bar{x} \vee \bar{y} \vee z, x, \bar{x}\}$) and $\dcsp(\{x \vee y \vee \bar{z}, x, \bar{x}\})$: given an instance $I$ of either problem, we obtain an equivalent instance of the other problem by replacing every literal $\ell$ with $\bar \ell$. Satisfying assignments are converted by replacing $0$-s with $1$-s and vice versa.
\end{proof}

As $\{x \vee y \vee \bar{z}, x, \bar{x}\}$ (resp.,~$\{\bar{x} \vee \bar{y} \vee z, x, \bar{x}\}$) is an irredundant base of $\mathrm{IV_2}$ (resp., $\mathrm{IE_2}$), Corollary~\ref{cor:irredbase} implies hardness if $\la \Gamma \ra$ contains $\mathrm{IV_2}$ or $\mathrm{IE_2}$.
\begin{corollary}\label{cor:IE2IV2}
  If $\Gamma$ is a (finite) constraint language with
  $\mathrm{IV_2}\subseteq \la \Gamma \ra$ or $\mathrm{IE_2}\subseteq \la \Gamma \ra$,
  then $\dcsp(\Gamma)$ is not FP-approximable, unless $\fpt = \WP$.
\end{corollary}

Finally, we consider the co-clone $\text{IN}_2$.
\begin{lemma}\label{IN2}
	If $\Gamma$ is a (finite) constraint language with $\mathrm{IN_2}\subseteq \la \Gamma \ra$ then $\dcsp(\Gamma)$ is not FP-approximable, unless $\mathrm{P = NP}$.
\end{lemma}
\begin{proof}
	Let $\nae = \{0,1\}^3 \setminus \{(0,0,0),(1,1,1)\}$. From Table~\ref{cc_table}, we see that $\nae$ is a base for the co-clone $\text{IN}_2$ of all self-dual languages. If there was a constant-factor \fpa\ algorithm for \dcsp(\nae), then setting the parameter to $0$ would give a polynomial time decision algorithm for \csp(\nae). But \csp(\nae) is NP-complete \cite{MR80d:68058}, so there is no constant-factor \fpa\ algorithm for \dcsp(\nae) unless $\mathrm{NP = P}$.
\end{proof}

\end{document}